\documentclass{article}

\usepackage[utf8]{inputenc}
\usepackage[english]{babel}
\usepackage[T1]{fontenc}
\usepackage{csquotes}
\usepackage[usenames,dvipsnames]{xcolor}

\usepackage{amsfonts,amsthm,amsmath,amssymb,mathtools}
\usepackage{algorithm, algpseudocode}
\usepackage{dsfont,newtxtext,newtxmath,microtype}

\usepackage{graphicx, quantikz, tikz,tikz-3dplot}
\usetikzlibrary{shapes.geometric}
\usetikzlibrary{positioning}
\usetikzlibrary{arrows}
\usetikzlibrary{decorations.pathmorphing}
\usetikzlibrary{calc, fit}

\usepackage{authblk}
\usepackage{booktabs,multirow}
\usepackage{braket}
\usepackage{cancel}
\usepackage[font=small]{caption}
\usepackage{comment}
\usepackage{diagbox}
\usepackage{enumerate}
\usepackage{float}
\usepackage{fullpage}
\usepackage[section]{placeins}
\usepackage{quantikz}
\usepackage{subfigure}
\usepackage{slashed}
\usepackage{todonotes}
\usepackage{bbm} 
\usepackage{yfonts}

\pgfdeclarelayer{background}
\pgfsetlayers{background,main}

\definecolor{navyblue}{rgb}{0.0, 0.0, 0.5}
\definecolor{LightPink}{rgb}{0.858, 0.188, 0.478}

\usepackage[colorlinks=true, urlcolor=blue,citecolor=purple,anchorcolor=blue, linkcolor=LightPink]{hyperref}

\newcommand{\acomm}{\alpha_\mathrm{comm}}

\newcommand{\Heff}{H_\mathrm{eff}}

\newcommand{\epsint}{\epsilon_\mathrm{ext}}

\newcommand{\epsdata}{\epsilon_\mathrm{data}}

\newcommand{\tF}{\tilde{F}}
\newcommand{\cP}{\mathcal{P}}

\newcommand{\cS}{\mathcal{S}}

\newcommand{\cO}{\mathcal{O}}
\newcommand{\sym}{\sigma}
\newcommand{\width}{\ell}
\newcommand{\amax}{a_\mathrm{max}}

\newcommand{\Dmax}{D_\mathrm{max}}
\newcommand{\Cexp}{C_\mathrm{Trot}}

\DeclarePairedDelimiter\abs{\lvert}{\rvert}
\DeclarePairedDelimiter\norm{\lVert}{\rVert}
\DeclarePairedDelimiter\expval{\langle}{\rangle}

\DeclareMathOperator{\ad}{ad}
\DeclareMathOperator{\tr}{tr}
\DeclareMathOperator{\polylog}{polylog}

\usepackage[capitalise]{cleveref}
\Crefname{lemma}{Lemma}{Lemmas}
\Crefname{proposition}{Proposition}{Propositions}
\Crefname{definition}{Definition}{Definitions}
\Crefname{theorem}{Theorem}{Theorems}
\Crefname{conjecture}{Conjecture}{Conjectures}
\Crefname{corollary}{Corollary}{Corollaries}
\Crefname{example}{Example}{Examples}
\Crefname{section}{Section}{Sections}
\Crefname{appendix}{Appendix}{Appendices}
\Crefname{figure}{Fig.}{Figs.}
\Crefname{equation}{Eq.}{Eqs.}
\Crefname{table}{Table}{Tables}
\Crefname{item}{Property}{Properties}
\Crefname{remark}{Remark}{Remarks}
\Crefname{fact}{Fact}{Facts}

\newtheorem{theorem}{Theorem}

\newtheorem{corollary}[theorem]{Corollary}

\newtheorem{lemma}[theorem]{Lemma}

\usepackage{ltablex}

\newcommand\prob\textsc
\makeatletter
\newcommand{\probleminput}[1]{\gdef\@probleminput{#1}}
\newcommand{\problemquestion}[1]{\gdef\@problemquestion{#1}}
\newcommand{\problempromise}[1]{\gdef\@problempromise{#1}}
\DeclareDocumentEnvironment{problem}{}{
\probleminput{}\problempromise{}\problemquestion{}
\leavevmode
}{
  \par\addvspace{0\baselineskip}
  \noindent
  \begin{itemize}
      \item[\textbf{Input:}] \@probleminput
\ifx\@problempromise\empty%
\else%
      \item[\textbf{Promise:}] \@problempromise
\fi%
      \item[\textbf{Output:}] \@problemquestion
  \end{itemize}
}
\makeatother

\usepackage[backend=biber,style=alphabetic,doi=false,isbn=false,url=false,maxbibnames=7,maxcitenames=2]{biblatex}
\bibliography{references}
\newbibmacro{string+doi}[1]{\iffieldundef{doi}{#1}{\href{https://dx.doi.org/\thefield{doi}}{#1}}}
\DeclareFieldFormat{title}{\usebibmacro{string+doi}{\mkbibemph{#1}}}
\DeclareFieldFormat[article]{title}{\usebibmacro{string+doi}{\mkbibquote{#1}}}
\DeclareFieldFormat[incollection]{title}{\usebibmacro{string+doi}{\mkbibquote{#1}}}
\title{Exponentially Reduced Circuit Depths Using Trotter Error Mitigation}

\author[1,2]{\href{https://orcid.org/0000-0002-6077-4898}{James~D.~Watson}}
\author[3]{\href{https://orcid.org/0000-0003-1478-7230}{Jacob~Watkins}}

\affil[1]{Joint Center for Quantum Information \& Computer Science, National Institute of Standards and Technology and University of Maryland}
\affil[2]{Department of Computer Science and Institute for Advanced Computer Studies, University of Maryland}
\affil[3]{
Facility for Rare Isotope Beams, Michigan State University}

\date{}

\begin{document}

{\begingroup
		\hypersetup{urlcolor=navyblue}
\maketitle
		\endgroup}

\begin{abstract}
    Product formulae are a popular class of digital quantum simulation algorithms due to their conceptual simplicity, low overhead, and performance which often exceeds theoretical expectations. 
    Recently, Richardson extrapolation and polynomial interpolation have been proposed to mitigate the Trotter error incurred by use of these formulae.  
    This work provides an improved, rigorous analysis of these techniques for the task of calculating time-evolved expectation values. 
    We demonstrate that, to achieve error $\epsilon$ in a simulation of time $T$ using a $p^\text{th}$-order product formula with extrapolation, circuits depths of $O\left(T^{1+1/p}\polylog(1/\epsilon)\right)$ are sufficient --- an exponential improvement in the precision over product formulae alone. 
    Furthermore, we achieve commutator scaling, improve the complexity with $T$, and do not require fractional implementations of Trotter steps.  
    Our results provide a more accurate characterisation of the algorithmic error mitigation techniques currently proposed to reduce Trotter error.
   
\end{abstract}

\newpage
\tableofcontents
\newpage

\section{Introduction} \label{sec:introduction}
\yinipar{Q}uantum simulation --- the task of computing dynamical properties of a quantum system --- has been an early inspiration and impetus for quantum computing, and is among the most promising candidates for near-term quantum advantage. 
Scientific domains such as quantum chemistry, nuclear physics, materials science, and high energy physics stand to benefit from robust, programmable quantum devices which can implement a quantum Hamiltonian of interest \cite{cao2019quantum,shaw2020quantum,roggero2020quantum, watson2023quantum, clinton2024towards}. 
Beyond the study of quantum phenomena, Hamiltonian simulation forms a key component to more generic routines such as linear systems solvers~\cite{harrow2009quantum, clader2013preconditioned}, and even the investigation of non quantum phenomena~\cite{costa2019quantum, babbush2023exponential}. 

Hamiltonian simulation is, by now, a relatively mature subfield of quantum computing. Significant attention has been given to all steps of the simulation procedure: qubit mappings~\cite{steudtner2018fermion,derby2021compact, wang2023quantum}, state preparation~\cite{ciavarella2022preparation}, time evolution~\cite{lloyd1996universal,berry2015simulating,low2019hamiltonian}, and measurement~\cite{knill2007optimal, somma2019quantum}.
Time evolution has received particular attention as it is often the most expensive step in a full routine. 
Roughly speaking, there currently exist four families of time evolution algorithms: product formulae (also known as Trotterization)~\cite{lloyd1996universal,childs2021theory}, linear combination of unitaries (LCU)~\cite{childs2012LCU, haah2021quantum,low2019well,aftab2024multi}, quantum walks \cite{berry2012black}, and qubitization~\cite{low2019hamiltonian}.
Despite the inferior asymptotic scaling of product formulae, they have many desirable properties, including conceptual and practical simplicity, lack of auxiliary qubits, natural incorporation of Lieb-Robinson bounds, commutator scaling, and the tendency to conserve desirable properties and symmetries of the Hamiltonian~\cite{tran2020destructive,childs2021theory, tran2021faster, csahinouglu2021hamiltonian, zhao2022hamiltonian, zhao2024entanglement}.
Many of these desirable properties are believed not to hold for post-Trotter methods in general \cite{zlokapa2024hamiltonian}.
Remarkably, the empirical performance of product formulae is often comparable to the LCU and qubitization methods in numerical studies~\cite{babbush2015chemical,childs2018toward}, and far better than leading error bounds predict~\cite{heyl2019quantum}. 
These studies indicate a gap in our theoretical understanding of expected Trotter error. 

While the unexpectedly high empirical performance of product formulae is fortunate, the current constraints on quantum hardware make simulating large, complicated Hamiltonians mostly out of reach, motivating the search for improved Trotter-based approaches that do not significantly increase quantum resources.
A wide range of product formulae algorithms have been suggested to optimise performance according to properties of the Hamiltonian under consideration~\cite{yuan2019theory,campbell2019random,ouyang2020compilation,morales2024greatly, nakaji2024qswift,sharma2024hamiltonian, bosse2024efficient, chen2024adaptive} and by optimising the circuit itself \cite{mckeever2023classically}.
Additional progress has been made by recognising that time evolution is not a full algorithm, but ultimately a subroutine embedded within a measurement protocol.
With suitable choice of measurements, the resulting classical data can be processed to improve results without going beyond product formulae. 
This amounts to algorithmic error mitigation, which is conceptually similar to the hardware error mitigation being developed and employed on current devices. 

Several Trotter mitigation approaches have been proposed, including Richardson extrapolation~\cite{endo2019,vazquez2023well}, polynomial interpolation~\cite{rendon2024improved} and parametric matrix models~\cite{cook2024parametric}.
The Richardson extrapolation and polynomial interpolation techniques are remarkably simple. 
Both involve taking an observable of interest and computing its time-evolved expectation value under the approximate Trotterized evolution for different time step sizes.
It is then possible to extrapolate to the zero step-size limit, corresponding to the perfect (i.e. un-Trotterized) time evolution. 

In this work, we conduct a rigorous and unified performance analysis of the Richardson extrapolation and polynomial interpolation methods, taking inspiration from a recent treatment by~\citeauthor{aftab2024multi} for multiproduct formulas~\cite{aftab2024multi}. 
In particular, we demonstrate the expected commutator scaling and $O(T^{1+1/p})$ scaling with the simulation time for a $p$th-order formula.
Our analysis separately considers "coherent" and "incoherent" measurement protocols for acquiring the expectation values for various Trotter step sizes. 
While the incoherent scheme optimises for short circuit depths, the coherent scheme achieves Heisenberg-limited precision and overall fewer quantum operations.
For the short-depth method, we show that only circuit depths of 
\begin{align*}
    O\left( (\lambda T)^{1+1/p}\polylog\left(1/\epsilon\right)\right)
\end{align*}
are sufficient, where $\lambda$ is a factor depending on commutators of terms in the Hamiltonian.
For the method which optimises asymptotic resources, we show that the overall number of Trotter steps which need to be implemented scales as
\begin{align*}
    O\left( \frac{(\lambda T)^{1+1/p}}{\epsilon}\polylog\left(1/\epsilon\right)\right).
\end{align*}
Essential to achieving these results is the fact that our extrapolation approaches are well-conditioned, such that small errors in the data do not rapidly accumulate. Finally, the algorithms require no additional control gates or ancillary qubits compared to regular product formulae.

The rest of this paper is outlined as follows. 
In~\Cref{sec:summary-of-methods}, we provide background on product formulae, the variation of parameters formula, extrapolation techniques, and measurement protocols. 
The main results are stated in~\Cref{sec:results}. 
\Cref{Sec:error-analysis} contains the technical error analysis of time-evolved observables under product formulae, which we then use to inform the complexity analyses of Richardson extrapolation (\Cref{sec:richardson-extrapolation}) and polynomial interpolation (\Cref{sec:polynomial-interpolation}). We support these theoretical findings with small numerical implementations in~\Cref{sec:numerical-demonstration}. \Cref{sec:classical_shadows} augments our approach with  the framework of classical shadows to estimate many time-evolved observables efficiently. Finally, we provide some discussion and concluding remarks in~\Cref{sec:discussion-and-conclusions}.

\section{Summary of Methods} \label{sec:summary-of-methods}
In terms of the algorithms considered in this work, the primary elements are product formulae, amplitude measurement, and two extrapolation techniques: Richardson and polynomial. For the theoretical analysis, we rely heavily on the variation of parameters formula from the theory of first-order ordinary differential equations. This section briefly reviews all of these components to help the reader understand the main results and subsequent proofs. Those interested only in result statements are welcome to skip to Section~\ref{sec:results}.

\subsection{Product Formulae}
Consider a (time independent) Hamiltonian $H$ expressed as a sum of $\Gamma$ terms
\begin{equation}
    H = \sum_{\gamma = 1}^\Gamma H_\gamma
\end{equation}
for Hermitian $H_\gamma$. Product formulae are splittings of the exponential $e^{-i H t}$, the time evolution operator, along the various terms $H_\gamma$. For example, the simplest product formula, namely first order Trotter, is defined by
\begin{equation}
    \cP_1(t) \coloneqq \prod_{\gamma=1}^\Gamma e^{-iH_\gamma t}
\end{equation}
where, by convention, we take the product going right to left. The utility of product formulae arises from the fact that there often exists decompositions of $H$ into terms $H_\gamma$ such that each exponential $e^{-i H_\gamma t}$ may be computed efficiently. For example, $H_\gamma$ may be $k$-local, or $H$ may be sparse and thus decomposable into 1-sparse terms. 

Product formulae are meant to approximate the exact time evolution operator for short times $t$. Larger times can be approximated to arbitrary precision by breaking the simulation time $T\in \mathbb{R}$ into sufficiently many steps. A product formula $\cP$ is said to be \emph{order} $p$ if
\begin{equation}
    \cP(t) - e^{-i H t} = O(t^{p+1})
\end{equation}
for small $t$. Thus, for time $T \in \mathbb{R}$ and $r \in \mathbb{Z}_+$,
\begin{equation}
    e^{-i H T} - \cP(T/r)^r = O(T^{p+1}/r^p).
\end{equation}
In this work, we only consider formulae of order at least 1. The order is roughly a proxy for accuracy, although large constant factors typically negate the advantage of high-order formulae in typical instances. There exist product formulae of arbitrarily large order, the most well-known family of which is the Trotter-Suzuki formulae $S_{2k}$ of order $p = 2k$. These are defined recursively as follows. For $k = 1$, 
\begin{align*}
    S_2(t) \coloneqq  \prod_{\gamma=\Gamma}^1 e^{-iH_\gamma t/2} \prod_{\gamma=1}^\Gamma e^{-iH_\gamma t/2}
\end{align*}
and for $k\in \mathbb{Z}_+$ greater than $1$,
\begin{align*}
    S_{2k}(t)\coloneqq [S_{2(k-1)}(u_kt)]^2 S_{2(k-1)}((1-4u_k)t)[S_{2(k-1)}(u_kt)]^2
\end{align*}
with a value of $u_k \in \mathbb{R}$ that is presently unimportant. One useful property of $S_{2k}$ is that it is symmetric, meaning $S_{2k}(-t) = S_{2k}^{-1}(t)$.

We will find it useful to define $s\coloneqq 1/r$ and treat $s$ as a continuous variable. It can be shown that $\cP(sT)^{1/s}$ is analytic in a neighbourhood of $s = 0$~\cite{rendon2024improved}. In such a neighbourhood, the product formula can be written as an evolution under an effective Hamiltonian
\begin{align} \label{eq:Delta_Heff}
\begin{aligned}
    \cP^{1/s}(sT) &= e^{-iT\Heff(sT)} \\
    \Heff(sT) &= H + \sum_{j\geq p} E_{j+1} s^jT^j \\
     &= H + E(sT)
\end{aligned}
\end{align}
where $E(t) \coloneqq H_\mathrm{eff}(t) - H$ and $E_{j+1}$ are a set of coefficients with explicit form given by the Baker-Campbell-Haussdorf (BCH) formula. When $\cP$ is symmetric, the error series is even in $s$, meaning $E_{j+1} = 0$ for odd $j$.

\subsection{Variation of Parameters Formula}

Our primary technical tool in this work will be the variation of parameters formula for linear operators. Given two matrices $A, B$, the formula reads
\begin{equation} \label{eq:var_of_param_general}
    e^{(A+B) t} = e^{A t} + \int_0^t e^{A(t - \tau)} B e^{(A+B) \tau} d\tau.
\end{equation}
This can be understood as a special case of the formula from the theory ordinary first-order differential equations, treating $B$ as a perturbation~\cite{aftab2024multi}.

We apply this formula by focusing on the relevant linear operators acting on the space of observables (sometimes dubbed "superoperators"). 
The exact and Trotterized evolutions of an observable $O$ under $H$ are $O(T) = e^{iHT}Oe^{-iHT}$ and $\tilde{O}(T,s) = e^{iT(H+E(sT))}Oe^{-iT(H+E(sT))}$, respectively, with $E$ defined in \cref{eq:Delta_Heff}.
We see that these are solutions to the first-order differential equations
\begin{equation}
    \partial_T O(T) = i\ad_{H}O(T) \qquad  \partial_T \tilde{O}(T,s) = i\ad_{H+E(sT)}\tilde{O}(T,s)
\end{equation}
where $\ad_{H}(\cdot) \coloneqq [H,\cdot]$. Taking $A = i \ad_H$ and $B = i \ad_{E(sT)}$, the formula~\eqref{eq:var_of_param_general} reads
\begin{align*}
   \tilde{O}(T,s) &= e^{iT\ad_{\Heff}}(O)   
   \\
   &= e^{iT\ad_{H}}(O) + i\int_{0}^Td\tau e^{i(T-\tau)\ad_{H}} \ad_{E(sT)}(\tilde{O}(\tau,s)).
\end{align*}
We can recursively insert the variation of parameters formula into the right hand side. Iterating this $K-1$ times, for a value of $K$ we will choose appropriately, we get a series expansion for the time-evolved observable in terms of the inverse Trotter-step $s$.
\begin{align}\label{Eq:Series_Expansion}
    \tilde{O}(T,s) = O(T) + \sum_{j\geq p} s^j \tilde{E}_{j+1,K}(T)(O) + \tilde{F}_K(s,T)(O)
\end{align}
Here $\tilde{E}$ and $\tilde{F}$ are (super)operator coefficients. 
The $F_K$ term can be thought of a type of remainder, made smaller with larger $K$, and
the coefficients $\tilde{E}_{j+1,K}$ consist of nested commutators of the $H_\gamma$. 
These claims are made rigorous in \cref{Sec:Error_Expansion}. The value of this expansion is that, by increasing $K$, we can make the remainder term irrelevant compared to the dominant errors captured by the power series in $\tilde{E}_{j+1, K}$. 
Thus, by characterising the $\tilde{E}_{j+1,K}$, it becomes much easier to understand the effects of extrapolation on the resulting estimate.

\subsection{Extrapolation Methods}
\Cref{Eq:Series_Expansion} defines a series expansion in terms of the inverse Trotter step size $s$ for the Trotterized time evolution of an observable. This, in turn, gives us a series expansion for the scalar-valued function
\begin{equation} \label{eq:Trotter_evolved_EVs}
    f(s) \coloneqq \expval{\tilde{O}(T,s)}.
\end{equation}
Ideally, we wish to estimate $f(0)$ to obtain a time evolved expectation value without Trotter error. By taking measurements at particular values of $s$, we can extrapolate to the $s=0$ ideal.
Here we review the two methods to perform this extrapolation considered in this work: Richardson extrapolation and polynomial interpolation.

\subsubsection{Richardson Extrapolation} \label{sec:richardson-background}
Given a function $f$ with a series expansion, Richardson extrapolation is a method to get iteratively improved estimates of a particular value \cite{richardson1911approximate, sidi_2003}.
Suppose we have a function with the expansion
\begin{align*}
    f(s) = \sum_{k=1}^\infty c_ks^k
\end{align*}
for some coefficients $c_k$, and we wish to approximate $f(0)$ when the only information about $f$ we have access to values of $f(s)$ for $s\in \{s_1,\dots , s_m\}$.
Then we can iteratively construct an approximation as follows.
Suppose we choose some $s_0 = s >0$ and $s_1 = x/k_1$ for integer $k_1 \neq 0$. 
We see that if we evaluate the function at a point $f(s/k_1)$, we have
\begin{align*}
     f(s/k_1) = f(0) + c_1\frac{s}{k_1} + c_2 \left(\frac{s}{k_1}\right)^2 + O(s^3).
\end{align*}
Since $f(s/k_1)$ and $f(s)/k_1$ match at first order, by subtracting the two equations yields
\begin{align*}
    f(s)/k_1 - f(s/k_1) = (1/k_1-1)f(0) + \tilde{c}_2' s^2 + O(s^3)
\end{align*}
or, after dividing by the factor multiplying $f(0)$,
\begin{align}
    \frac{f(s)/k_1 - f(s/k_1)}{1/k_1 - 1} \equiv F^{(1)}(s) = f(0) + \tilde{c}_2 s^2 + O(s^3).
\end{align}
Thus, we have constructed an estimator $F^{(1)}(s)$ of $f(0)$ accurate to order $O(s^2)$. By comparison, $f(s)$ is only order $O(s)$. The precise form of $\tilde{c}_2$ is irrelevant for the present illustration. The procedure can be iterated again to eliminate the $\tilde{c}_2$ terms, and so on as desired.

In general, a linear combination of $m$ evaluations of $f$, with suitably chosen coefficients $b_j$, removes $m$ terms in the series expansion. An $m$-term Richardson extrapolation $F^{(m)}(s)$ has a small $s$ behaviour which satisfies
\begin{align*} 
    |F^{(m)}(s) - f(0)| &\leq \sum_{j=1}^m |b_j| \left| \sum_{k=m+1}^\infty c_k s^k  \right| 
    \\
    &=  O(s^{m+1}).
\end{align*}
Thus, by picking sufficiently small $s$ and sufficiently large $m$, we achieve a more accurate estimate of $f(0)$.
We can then apply this to the series expansion we have from \cref{Eq:Series_Expansion}.

When applying a well-condition Richardson extrapolation for time evolved expectation values, we will find that the error scaling as
\begin{align*}
    |F^{(m)}(s) - \expval{O(T)}|= O(s^{2m}T^{2m(1+1/p)})
\end{align*}
for a symmetric order $p$ product formula. By choosing the sampling points $\{s_i\}_{i=1}^m$ such that $s_1$ is the largest, and all other are related by the Chebyshev nodes starting from $s_1$, one can also show robustness of the estimator $F^{(m)}(s)$ to noisy data~\cite{low2019well}. 
We give a more thorough analysis of Richardson extrapolation applied to Trotter-evolved expectation values in \cref{sec:richardson-extrapolation}.

\subsubsection{Polynomial Interpolation}
Here we review the proposal of \citeauthor{rendon2024improved} for using polynomial interpolation as a Trotter extrapolation method~\cite{rendon2024improved}. Given a real-valued function $f(s)$ in a domain $D \subseteq \mathbb{R}$, and given a set of values $f(s_1), \dots, f(s_m)$, there is a unique $(m-1)$ degree polynomial $P_{m-1} f$ which matches the value of $f$ at each $s_i$. A concrete representation may be given terms of the Lagrange basis polynomials
\begin{align*}
    P_{m-1} f(s) = \sum_{i=1}^m f(s_i) \mathcal{L}_i(s),
\end{align*}
where 
\begin{align*}
    \mathcal{L}_i(s) \coloneqq \prod_{\substack{1\leq n\leq m\\ n\neq i}} \frac{(s-s_n)}{(s_i-s_n)}.
\end{align*}
More importantly, $P_{m-1} f$ may be computed by a number of standard techniques such as barycentric interpolation, which is efficient and stable to floating-point errors~\cite{higham2004numerical}. By choosing the sample points $s_i$ to be at the \emph{Chebyshev nodes} of an interval of interest, the interpolation will be robust to errors in the computed values $f(s_i)$~\cite{rivlin2020chebyshev}. These features make interpolation a plausible tool for Trotter extrapolation.

To apply polynomial interpolation to Trotter-evolved expectation values $f(s)$ of~\Cref{eq:Trotter_evolved_EVs}, we consider a neighbourhood $[-\width,\width]$ of the origin, where $\width\in \mathbb{R}_+$ is to be chosen based on problem parameters. This function can be computed for both positive and negative $s$ using product formulae evolutions on a quantum computer. Performing a Chebyshev interpolation of $m\in 2\mathbb{Z}_+$ sample points, the final estimate is given by $P_{m-1}f(0)$.~\Cref{Fig:Polynomial_Interpolation} depicts this setup. One can show~\cite{rendon2024improved} that the interpolation error at the point $s=0$ is bounded as
\begin{align}\label{Eq:Polynomial_Approx_Error_Intro}
     |f(0) - P_{m-1}f(0)| \leq \max_{\xi \in [-\width,\width]} \abs{f^{(m)}(\xi)} \left(\frac{\width}{2m}\right)^{m}.
\end{align}
Utilising the error expansion~\cref{Eq:Series_Expansion}, we can bound the derivative term as
\begin{align*}
     \partial_s^m \expval{\tilde{O}(T,s)} = O\left(\width^{m}T^{m(1+1/p)}\right).
\end{align*}
Hence, by choosing the interval with $\width = O(T^{-(1+1/p)})$, we get that the approximation error from~\cref{Eq:Polynomial_Approx_Error_Intro} decreasing exponentially with $m$.
See~\cref{sec:polynomial-interpolation} for the complete analysis.
\begin{figure}
    \centering
    \includegraphics[width=1.0\textwidth]{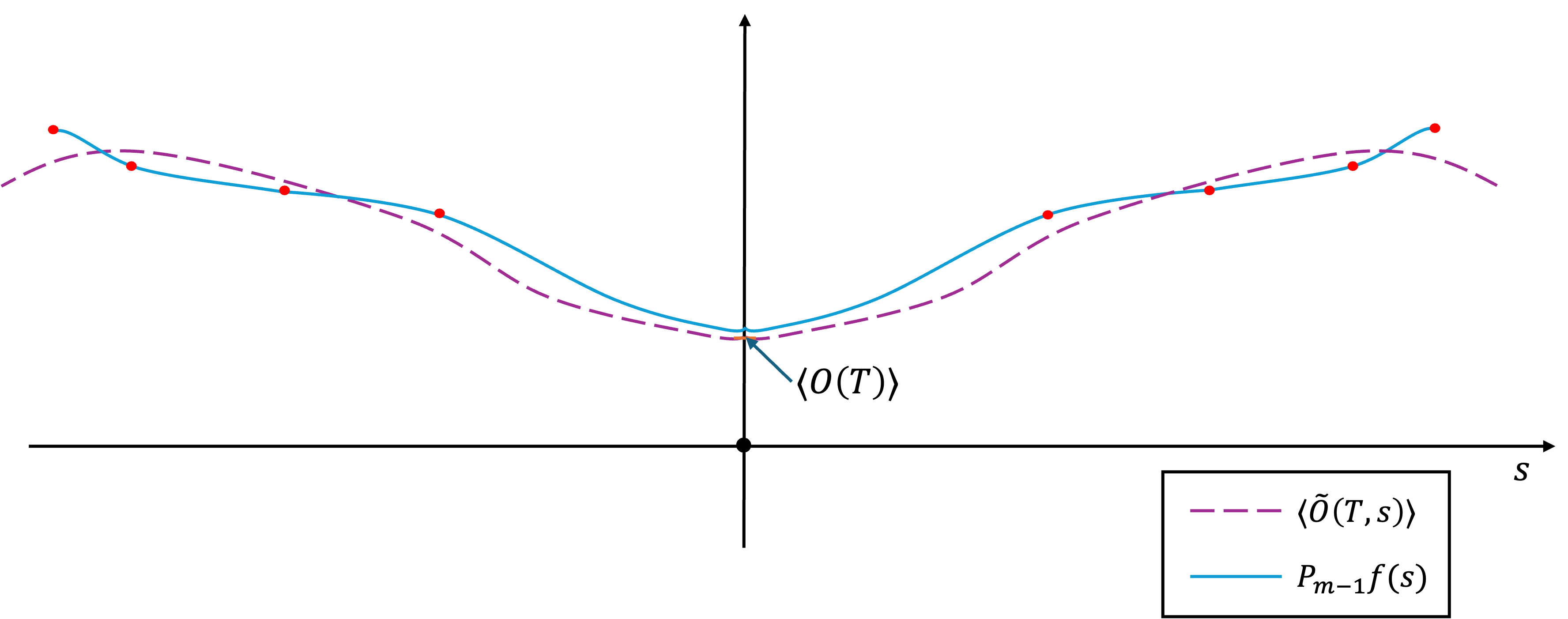}
    \caption{Schematic of the polynomial interpolation procedure. The dotted purple line is the true value of the  Trotter-evolved expectation value, $\langle \tilde{O}(T,s)\rangle$, and the blue line is the interpolating polynomial $P_{m-1} f(s)$ for $m=8$. The red points are estimates of the time-evolved expectation value obtained via product formula evolution with measurement, which are then used to construct the polynomial. The final estimator of the expectation value is given by $P_{m-1} f(0)$.}
    \label{Fig:Polynomial_Interpolation}
\end{figure}

\subsection{Taking the Measurements} \label{sec:measurements}
So far, we have been concerned with the extrapolation schemes and their accuracies.
However, to implement either Richardson extrapolation or polynomial interpolation, we need to be able to take measurements of the time-evolved observable at different value of $s$, and these measurements come with intrinsic errors (regardless of hardware effects). For successful overall error reduction, we will require our measurements to be within some error tolerance of the exact \emph{Trotter}-evolved expectation value, and from robustness guarantees, this allows us to extrapolate to within a small error of the \emph{ideal} evolution. In this work, we consider two schemes for performing the expectation value measurements.
\ \newline 
\ \newline
\noindent \textbf{Method 1 (Incoherent):}
Simply time-evolve the initial state $\rho_0$ under $\cP^{1/s_j}(s_j T)$, then measure the expectation value by repeated measurement of $O$. 
The precision will necessarily be shot noise limited, leading to a $O(1/\epsilon^2)$ cost for precision $\epsilon$. 
However, circuit depths will only be as long as needed to perform a single measurement, and often these can be estimated straightforwardly with low depth circuits (e.g. using classical shadows using randomised single-qubit measurements).
\ \newline
\ \newline 
\noindent \textbf{Method 2 (Coherent):}
Use iterative quantum amplitude amplification to achieve Heisenberg-limited scaling $\tilde{O}(1/\epsilon)$ in the measurement precision.
This is essentially a quadratic improvement over in error scaling Method 1 but involves longer circuit depths. Iterative protocols for amplitude estimation require few auxiliary qubits while still achieving competitive results compared to Fourier-based methods. For our analysis, we consider specifically the iterative Quantum Amplitude Estimation protocol of~\citeauthor{grinko2021iterative}~\cite{grinko2021iterative}.
\ \newline
\ \newline 
A flowchart illustrating the entire extrapolation process is given in \cref{Fig:Flowchart}.
Method 1 optimises for shorter circuit depths at the expense of greater overall resources by using incoherent measurements for observables. As such, this method may be more promising for NISQ devices. Method 2 uses coherent measurements to improve the overall scaling at the expense of longer circuits, and hence may be more useful for fault-tolerant devices. 
\begin{figure}
    \centering
    \includegraphics[width=0.8\textwidth]{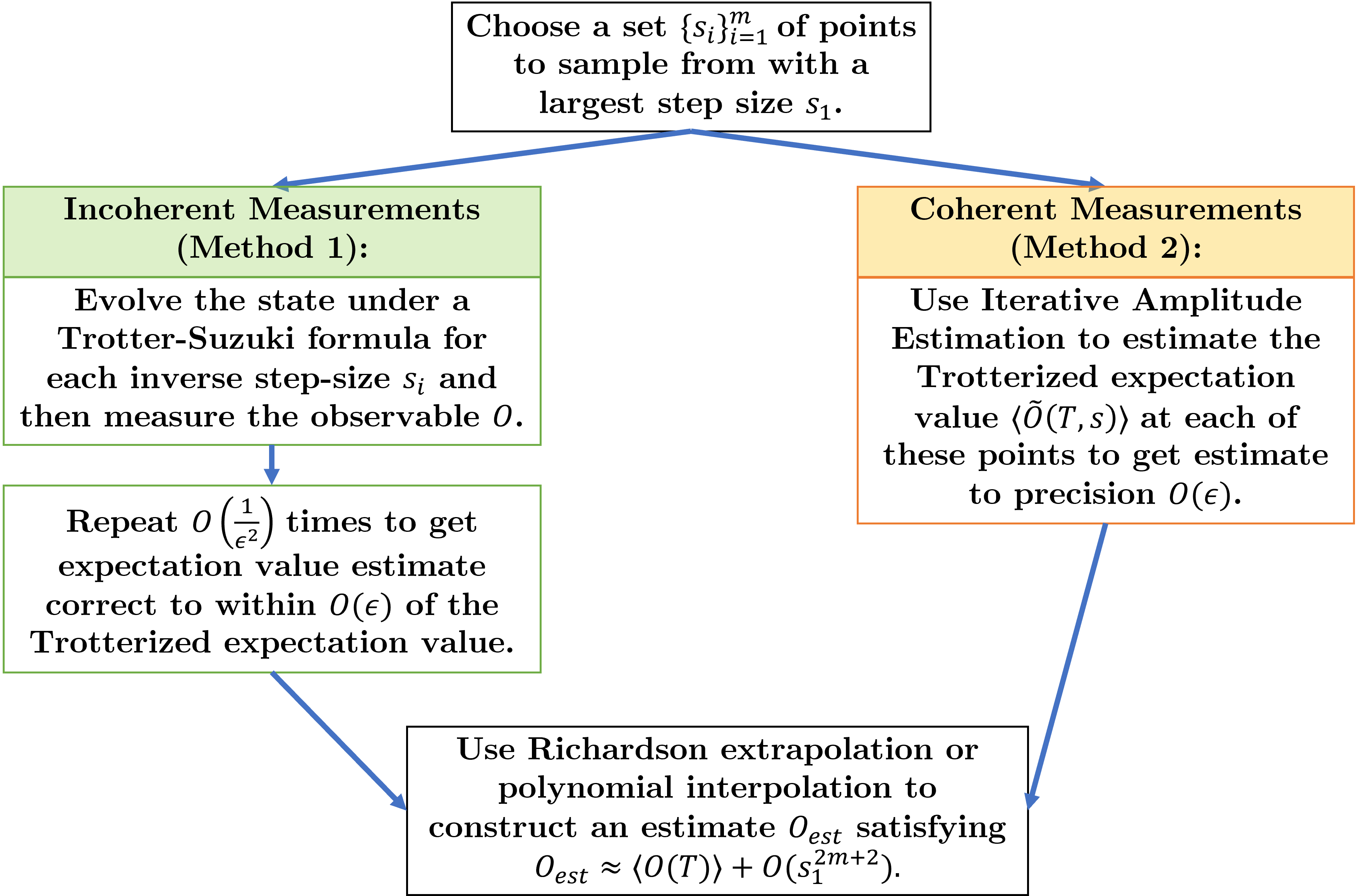}
    \caption{Flowchart indicating the incoherent (Method 1) and coherent (Method 2) schemes for performing extrapolation of the time-evolved expectation values using product formulae, where the specified time and final error are $T$ and $\epsilon$ respectively. For each time-step size, we need to make a measurement of the observable's expectation to precision $\epsilon$. While the incoherent scheme is simpler and requires shorter circuit depths, the coherent scheme achieves a quadratic speedup from using quantum amplitude estimation. Regardless of the measurement protocol, the algorithm concludes with the same classical calculation of the acquired data. }
    \label{Fig:Flowchart}
\end{figure}

\section{Results} \label{sec:results}

Our results are stated for two separate measurement routines for performing the expectation value calculations, as discussed in~\cref{sec:measurements}. \cref{Fig:Flowchart} gives a schematic workflow for the extrapolation protocols using either measurement type. We now informally state the main results of our work; see later sections for more rigorous formulations.

\begin{theorem}[Trotter Extrapolation Resource Counts (Informal)]
\label{Theorem:Informal_Main_Theorem}
     Let $O$ be an observable and $H = \sum_\gamma^\Gamma H_\gamma$ be a time independent Hamiltonian.
    Let $O_{\text{est}}$ be the estimate for a time-evolved expectation value on an arbitrary initial state, $\expval{O(T)}$, produced by varying the Trotter step-size of a $p^\text{th}$-order Trotter-Suzuki formula, taking $m$ measurement samples, and then extrapolating to the zero step-size limit using either Richardson extrapolation or polynomial interpolation.
    Then, for a simulation time $T$, a relative error $\epsilon$ can be achieved, with high probability, such that
    \begin{align*}
        |\expval{O(T)} - O_{\text{est}}|\leq \epsilon \norm{O}
    \end{align*}
    using $m=O(\log(1/\epsilon))$ extrapolation points. 
    Furthermore, the maximum circuit depth and total gate count for each of the estimation subroutines described above scales as given in the following table.
     \begin{table}[H]
        \centering
        \begin{tabular}{c||c|c}
           \diagbox{Method}{Cost}  & \textbf{Max Circuit Depth} & \textbf{Total Gate Cost} \\
           \hline\hline
           \textbf{Incoherent}    & $\tilde{O}\left( (\lambda T)^{1+1/p}\polylog\big( \frac{1}{\epsilon}\big) \right)$  & $\tilde{O}\left( \frac{(\lambda T)^{1+1/p}}{\epsilon^2}\polylog\big( \frac{1}{\epsilon}\big) \right)$ \\
           \hline 
           \textbf{Coherent}& $\tilde{O}\left( \frac{(\lambda T)^{1+1/p}}{\epsilon}\polylog\big( \frac{1 }{\epsilon}\big) \right)$ & $\tilde{O}\left( \frac{(\lambda T)^{1+1/p}}{\epsilon}\polylog\big( \frac{1}{\epsilon}\big) \right)$ \\
        \end{tabular}
        \label{Table:Resource_Costs_Main_Results}
    \end{table}
    
    \noindent Here $\tilde{O}$ hides $\log\log$ factors, and $\lambda \leq 4\sum_\gamma \norm{H_\gamma}$ is a nested commutator with full expression given in \cref{Lemma:Sufficient_Trotter_Depth}.
\end{theorem}

\noindent 
These results demonstrate an improvement in circuit depth over the direct application of product formulae by an exponential in the error scaling, while maintaining the same time scaling
(see  \cref{Table:Resource_Costs_Regular_Trotter} for a direct comparison for the results in  \cref{Theorem:Informal_Main_Theorem}). As we discuss in \cref{sec:Previous_Literature} below, our analysis improves on prior results for both Richardson extrapolation and polynomial interpolation.
The extrapolated techniques are shown to inherit all the of desirable properties of standard product-formulae methods such as locality, commutator scaling, respecting the system's symmetries, etc.
We note that the results have either $1/\epsilon$ or $1/\epsilon^2$ scaling for the coherent or incoherent measurement procedures respectively; this is a consequence of the measurement protocols, not the time evolution. 
Indeed, this scaling would be present even for LCU or quantum signal processing simulation techniques using the incoherent or coherent measurement protocols, and lower bounds can be shown from results in metrology \cite{giovannetti2006quantum}.

 \begin{table}[H]
        \centering
        \begin{tabular}{c||c|c}
        \multicolumn{3}{c}{\textbf{Trotter Performance without Extrapolation}} \\
           \hline \diagbox{Method}{Cost}  & \textbf{Max Circuit Depth} & \textbf{Total Gate Cost} \\
           \hline\hline
           \textbf{Incoherent} \cite{childs2021theory}  & $\tilde{O}\left(  \frac{(\acomm^{(p+1)})^{1/p} T^{1+1/p}}{\epsilon^{1/p}} \right)$  & $\tilde{O}\left( \frac{(\acomm^{(p+1)})^{1/p} T^{1+1/p}}{\epsilon^{2+1/p}} \right)$ \\
           \hline 
           \textbf{Coherent} \cite{childs2021theory} & $\tilde{O}\left(  \frac{(\acomm^{(p+1)})^{1/p} T^{1+1/p}}{\epsilon^{1+1/p}} \right)$ & $\tilde{O}\left(  \frac{(\acomm^{(p+1)})^{1/p} T^{1+1/p}}{\epsilon^{1+1/p}} \right)$ \\
        \end{tabular}
        \caption{Summary of asymptotic scalings for Trotter-based protocols for time-evolved expectation values without extrapolation. Results are obtained using current best scalings for product formulae (see reference). This should be compared to the results in \cref{Theorem:Informal_Main_Theorem}. The quantity $\acomm^{(p+1)} \leq 2^p (\sum_\gamma \norm{H_\gamma})^{p+1}$ is a sum of $p+1$ nested commutators defined in \cref{Eq:acomm_definition}.}
        \label{Table:Resource_Costs_Regular_Trotter}
    \end{table}

Additionally, method 1 (using incoherent measurements) is fully compatible with classical shadow-based techniques. 
That is, if we wish to estimate $M$ different local observables with probability $\geq 1- \delta$, then we get an additional resource cost $O(\log(M/\delta))$.
In~\cref{sec:Symmetries}, we also show how the prefactors in \cref{Theorem:Informal_Main_Theorem} can be improved by taking into account symmetries of the system.

Finally we note that the bounds in \cref{Theorem:Informal_Main_Theorem} are in terms of the number of Trotter steps. 
The total number of elementary exponential operations comes with a multiplicative factor $\Upsilon\Gamma$ which may contain additional information about scaling with system size. A true elementary circuit depth and gate cost would require more specification of a computational model, such as a $k$-local or sparse matrix model.

\subsection{Comparison to Previous Extrapolation Results}
\label{sec:Previous_Literature}

\paragraph{Richardson Extrapolation.} The use of Richardson extrapolation to reduce Trotter error was originally proposed in \cite{endo2019}. Although the authors note that Richardson extrapolation should reduce the error with repeated samples, the work does not provide rigorous resource estimates in terms of the basic simulation parameters.
Subsequent work in \cite{vazquez2023well} employs Richardson extrapolation on actual quantum hardware, observing improvements over bare simulation results. 
However, there appear to be errors in the formal treatment of the algorithmic error, which we discuss in \cref{sec:Errata}.   
As a result, a rigorous asymptotic analysis of the scaling of Richardson extrapolation for quantum simulation does not appear to be present in the literature.
We hope the present work provides a more complete theoretical picture of the Richardson approach.

\paragraph{Polynomial Interpolation.} The polynomial interpolation approach was introduced and rigorously analysed in \cite{rendon2024improved}. There it was shown that, to measure an observable at time $T$, it is sufficient to use a total number of Trotter steps (in the coherent measurement case) scaling as $\tilde{O}\left(\frac{(\chi T)^2}{\epsilon}\log\left(\frac{T}{\epsilon}\right) \right)$, \textit{independent of the order of product formula}, where $\chi = \Gamma \max_\gamma \norm{H_\gamma}$.
Besides scaling as a first-order product formula in $T$, regardless of the order, it additionally does not demonstrate the expected commutator scaling, and instead depends on the Hamiltonian's norm.
Furthermore, the prior work invokes the use of fractional implementations of the product formula, i.e. $\cP^{\delta}$ for $\delta \in (0,1)$.
Although this can be implemented by quantum signal processing methods, this may not be easily achievable on a NISQ device.
If the fractional implementation is not used, then previous analysis leads to a $O\left( \frac{ T^3}{\epsilon}\log\left(\frac{T}{\epsilon}\right) \right)$ scaling on account of the imperfect Chebyshev nodes. 
By contrast, the present work demonstrates improved time scaling, commutator scaling, and avoids the need to implement fractional product formulae.
Recent work in \cite{rendon2023towards} demonstrates how to deal with the fractional queries and achieves a $O(T^{1+1/p})$ scaling in the time parameter, but does not give explicit prefactors or commutator scaling.

\paragraph{Other Approaches using Classical Post-Processing.}

Other approaches at ``dequantising'' Hamiltonian simulation techniques exist, but are not directly related to the present work. For example, \cite{faehrmann2022randomizing} introduce a randomised version of the multiproduct formula presented in \cite{low2019well}, achieving near-optimal scaling. 
However, the method requires auxiliary qubits and control operations.
Further work on the dequantising the multiproduct formula can be seen in  \cite{zhuk2023trotter}, in which they improve the scaling of the multiproduct formula and show it satisfies commutator scaling with a quadratic improvement in the error associated with the size of the time-steps relative to standard $p^{th}$-order product formulae.
Other approaches include constructing non-unitary channels which can be computed using classical post-processing using randomised compiling methods \cite{gong2023improved, nakaji2024qswift}.

\section{Error Analysis}
\label{Sec:error-analysis}

In this section, we derive our primary technical results and develop our most important tool: an explicit series expansion, with respect to Trotter step, of a time evolved expectation value of an observable $O$ evolved under a product formula. The results contained here will be directly applied to Richardson extrapolation and polynomial interpolation in the appropriately labelled subsequent sections. Our approach is directly inspired by the recent work of~\citeauthor{aftab2024multi}~\cite{aftab2024multi}, and much of the logic follows closely. 

We begin with some essential concepts and definitions. A staged product formula $\cP$ is, in the sense of~\cite{childs2021theory}, one which can be expressed as
\begin{align}\label{Eq:TS_Ordering}
    \cP(t) \coloneqq \prod_{\upsilon=1}^\Upsilon \prod_{\gamma=1}^\Gamma e^{-i t a_{(\upsilon,\gamma) H_{\pi_\upsilon (\gamma)}}},
\end{align}
where $\Upsilon \in \mathbb{Z}_+$ is the number of \emph{stages}, $a_{(\upsilon,\gamma)} \in\mathbb{R}$ are coefficients and $\pi_\upsilon \in S_\Gamma$ are permutations. For $p \in\mathbb{Z}_+$, a $p$th order product formula satisfies
\begin{equation}
    \cP(t) - e^{-i H t} = O(t^{p+1})
\end{equation}
in the small $t$ limit. 
For ease of notation, let $[X_1 X_2 \ldots X_n]$ refer to a right-nested $n$-commmutator
\begin{equation}
    [X_1 X_2 \ldots X_n] \coloneqq [X_1,[X_2,[\ldots[X_{n-1},X_n]\ldots]]]
\end{equation}
and define
\begin{equation}\label{Eq:acomm_definition}
    \acomm^{(j)} \coloneqq \sum_{\gamma_1\gamma_2\ldots \gamma_j = 1}^\Gamma \norm{[H_{\gamma_1} H_{\gamma_2} \ldots H_{\gamma_j}]}.
\end{equation}
It will be useful, for our purposes, to employ a BCH formula based on right-nested commutators~\cite{arnal2021note}. We first define
\begin{equation} \label{eq:phi_def}
    \phi_j(Y_1, Y_2, \ldots, Y_j) \coloneqq \frac{1}{j^2} \sum_{\sigma \in S_j} (-1)^{d_\sigma} \binom{j-1}{d_\sigma}^{-1} [Y_{\sigma(1)}\ldots Y_{\sigma(j)}].
\end{equation}
where $S_j$ is the set of permutations over $j$ elements and $d_\sigma$ is the number of \emph{descents} in $\sigma$. An index $i$ is a descent of $\sigma$ if $\sigma(i) > \sigma(i+1)$. In these terms, the BCH formula reads
\begin{equation}
    \prod_{i=1}^n e^{X_i} = e^Z
\end{equation}
where 
\begin{equation} \label{eq:Z_BCH}
    Z = \sum_{i=1}^n X_i + \sum_{j=2}^\infty \frac{1}{j!} \sum_{\mathcal{J}} \binom{j}{j_1 \ldots j_n} \phi_j\left(X_1^{\times j_1}, \ldots, X_n^{\times j_n}\right).
\end{equation}
Here, the $\mathcal{J}$ sums appropriately over the multinomial
\begin{equation}
    \mathcal{J} \coloneqq \{(j_1,\ldots,j_n)\in \mathbb{N}^n \mid \sum_{j=1}^n j_i = j\}
\end{equation}
and $X^{\times i}$ refers to $i$ copies of $X$ in the argument. The expression~\cref{eq:Z_BCH} is formal and may not converge. Convergence is guaranteed provided that the nested commutators do not grow too rapidly. 

    

We begin with a lemma concerning an error series in the effective Hamiltonian $\Heff(t)$, which approximates $H$ for small $t$. It follows essentially from \cite[Theorem 9]{childs2021theory}.

\subsection{Error Expansion of the Trotterized Operator}
\label{Sec:Error_Expansion}

\begin{lemma}[Effective Hamiltonian Error Series]\label{Lemma:Effective_Hamiltonian_Error_Series}
    Let $\cP$ be a staged product formula with coefficients $a_{(\upsilon,\gamma)}$ of order $p \in \mathbb{Z}_+$, and let $\Heff$ be the effective Hamiltonian of $\cP$ defined by the relation
    \begin{align*}
        \cP(t) = e^{-it\Heff(t)}.
    \end{align*}
    for $t\in\mathbb{R}$. Suppose that there exists a $J\in \mathbb{Z}_+$ and $C\in\mathbb{R}_+$ such that
    \begin{align*}
        \sup_{j\geq J} \acomm^{(j)}\left(a_\mathrm{max}\Upsilon \abs{t}\right)^j\leq C,
    \end{align*}
    with $a_\mathrm{max} \coloneqq \max_{\upsilon,\gamma} \abs{a_{(\upsilon,\gamma)}}$.
    Then the effective Hamiltonian can be written as
    \begin{align*}
        \Heff(t) = H + \sum_{j = 1}^\infty E_{j+1} t^j 
    \end{align*}
    where 
    \begin{equation*}
        E_j \coloneqq \frac{(-i)^{j-1}}{j!} \sum_{\mathcal{J}} \binom{j}{j_1\ldots j_n}\left(\prod_{i=1}^n a_i^{j_i}\right) \phi_j\left(H_{\gamma_1}^{\times j_1},\ldots,X_{\gamma_n}^{\times j_n}\right)
    \end{equation*}
    and $n = \Upsilon \Gamma$. Moreover, $E_j$ satisfies the bound
    \begin{align*}
        \norm{E_j} \leq \frac{(a_\mathrm{max}\Upsilon)^j}{j^2}\acomm^{(j)}.
    \end{align*}
    .
\end{lemma}
\begin{proof}
    Consider $\cP$ written as
    \begin{equation}
        \cP(t) = \prod_{i=1}^n e^{-i t a_i H_i}
    \end{equation}
    with $H_i \equiv H_{\gamma_i}$ notated as such for simplicity, and because we aren't concerned with the possibility that $H_i = H_j$ may frequently occur. By definition of $\Heff$ we have
    \begin{equation}
        \prod_{i=1}^n e^{-i t a_i H_i} = e^{-i \Heff t}
    \end{equation}
    and may therefore express $\Heff$ as a formal BCH expansion. Using~\cref{eq:Z_BCH} with $Z = -i t \Heff (t)$ and $X_i = -i t a_i H_i$,
    \begin{equation}
        -i t \Heff(t) = \sum_{i=1}^n -i t a_i H_i + \sum_{j=2}^\infty \frac{1}{j!} \sum_{\mathcal{J}} \binom{j}{j_1\ldots j_n}\phi_j\left((-ita_1 H_1)^{\times j_1},\ldots, (-it a_n H_n^{\times j_n})\right).
    \end{equation}
    Since $\cP$ is at least $1$st order, we have $\sum_i a_i H_{\gamma_i} = H$. Using the multilinearity of $\phi_j$,
    \begin{align}
    \begin{aligned}
        \Heff(t) &= H + \sum_{j=1}^\infty \frac{(-i t)^j}{(j+1)!} \sum_{\mathcal{J}} \binom{j+1}{j_1\ldots j_n}\left(\prod_{i=1}^n a_i^{j_i}\right) \phi_{j+1}\left(H_1^{\times j_1},\ldots,H_n^{\times j_n}\right) \\
        &= H + \sum_{j=1}^\infty t^j E_{j+1}
    \end{aligned}
    \end{align}
    where we've defined the Hermitian error operators
    \begin{equation}
        E_j \coloneqq \frac{(-i)^{j-1}}{j!}\sum_{\mathcal{J}} \binom{j}{j_1\ldots j_n} \left(\prod_{i=1}^n a_i^{j_i}\right)\phi_j\left(H_1^{\times j_1},\ldots,H_n^{\times j_n}\right).
    \end{equation}
    Applying the triangle inequality and using the definition of $\phi_j$ in~\cref{eq:phi_def} 
    \begin{equation}
        \norm{E_j} \leq \frac{1}{j!} \sum_{\mathcal{J}} \binom{j}{j_1\ldots j_n} \left(\prod_{i=1}^n \abs{a_i}^{j_i}\right) \frac{1}{j^2} \sum_{\sigma \in S_j} \binom{j-1}{d_\sigma}^{-1} \norm{[H_{\sigma(i_1)}\ldots H_{\sigma(i_j)}]},
    \end{equation}
    where 
    \begin{equation}
        H_{i_1}\ldots H_{i_j} = H_1^{\times j_1}\ldots H_n^{\times j_n}
    \end{equation}
    and $i_k \in \{1,\ldots,n\}$. We wish to reindex this sum to be over tuples $(i_1, \ldots, i_j)$ with $i_k \in \{1,\ldots n\}$ varying freely, and accomplish this by rehashing arguments used surrounding~\cite[Eq. (44)]{aftab2024multi}. For a given $[H_{i_1}\ldots H_{i_j}]$, there exists a unique $(j_1,\ldots,j_n)\in \mathcal{J}$ specifying the terms in that commutator, as the $j_k$ indices give the number of each $H_k$ present in the commutator, and this is a function of a full $(i_1, \ldots, i_j)$ specification. Thus, we may write
    \begin{equation}
        \norm{E_j} \leq \frac{1}{j! j^2} \sum_{i_1\ldots i_j = 1}^n \norm{[H_{i_1}\ldots H_{i_j}]} \binom{j}{\vec{\jmath}} \left(\prod_{k=1}^j \abs{a_{i_k}} \right) \sum_{\sigma \in S(i_1,\ldots, i_j)} \binom{j-1}{d_\sigma}^{-1}
    \end{equation}
    where $\vec{\jmath} = (j_1,\ldots j_n)$ is the vector of counts determined by the $i_k$, and $S(i_1,\ldots, i_j) \subseteq S_j$ refers to the permutations which leave the sequence $(i_1,\ldots,i_j)$ invariant. There are exactly $j_1!j_2!\ldots j_n!$ of these. Thus,
    \begin{equation}
        \sum_{\sigma \in S(i_1,\ldots, i_j)} \binom{j-1}{d_\sigma}^{-1} \leq \sum_{\sigma \in S(i_1,\ldots, i_j)} 1 = j_1! j_2!\ldots j_n! 
    \end{equation}
    and therefore
    \begin{align}
    \begin{aligned}
        \norm{E_j} &\leq \frac{1}{j^2} \sum_{i_1,\dots i_j = 1}^n \norm{[H_{i_1}\ldots H_{i_j}]} \prod_{k=1}^j \abs{a_{i_k}} \\
        &\leq \frac{a_\mathrm{max}^j}{j^2} \sum_{i_1,\dots i_j = 1}^n \norm{[H_{i_1}\ldots H_{i_j}]} 
    \end{aligned}
    \end{align}
    where $a_\mathrm{max} \coloneqq \max_i a_i$. We now remember that each $H_{i_n}$ runs $\Upsilon$ times over each term $H_\gamma$ in $H$. Thus, each sequence $H_{\gamma_1} H_{\gamma_2}\ldots H_{\gamma_j}$ is represented $\Upsilon^j$ times in the multi-index $(i_1,\ldots,i_j)$. Hence,
    \begin{equation} 
        \sum_{i_1,\dots i_j = 1}^n \norm{[H_{i_1}\ldots H_{i_j}]} = \Upsilon^j \sum_{\gamma_1,\ldots,\gamma_j = 1}^\Gamma  \norm{[H_{\gamma_1}\ldots H_{\gamma_j}]} = \Upsilon^j \acomm^{(j)}.
    \end{equation}
    This gives our upper bound on $E_j$ from the lemma statement. 

    To summarise, we have that $t \Heff(t) = tH + \sum_{j=1}^\infty E_{j+1} t^{j+1}$ provided that the series converges. A sufficient condition is absolute convergence, namely the convergence of 
    \begin{equation}
        \sum_{j=2}^\infty \abs{t}^j \norm{E_j}.
    \end{equation}
    We note that $\sum_j C/j^2$ is a convergent series for all $C \in \mathbb{R}_+$, and thus, by the squeeze theorem, it suffices that there exists an $J \in \mathbb{Z}_+$ such that for all $j\geq J$, $t^j \norm{E_j} \leq C/j^2$. Using our bound on $\norm{E_j}$, this is satisfied provided
    \begin{equation}
        \left(a_\mathrm{max} \Upsilon t\right)^j \acomm^{(j)} \leq C
    \end{equation}
    for such $j$. This is equivalent to the condition provided in the lemma.
\end{proof}

Observe that less stringent conditions for convergence could be derived, e.g., by bounding the series with a $1/j^{1 + \epsilon}$ decay for any $\epsilon > 0 $ instead of $1/j^2$, and indeed the condition under which such series converge is well studied \cite{lakos2017convergence, lakos2019convergence}. 
However, our condition here is simple enough and not too stringent. Using the bound
\begin{equation} \label{eq:acomm_bound}
    \acomm^{(j)} \leq \frac{1}{2} \left(2 \sum_\gamma \norm{H_\gamma}\right)^j
\end{equation}
we may obtain the simpler, sufficient condition
\begin{equation}
    2 a_\mathrm{max}\Upsilon \abs{t} \sum_{\gamma = 1}^\Gamma \norm{H_\gamma} < 1
\end{equation}
which shows that, for any $H$, there exists an open neighbourhood about $t = 0$ for which the BCH series converges. For the case of Trotter-Suzuki formulae $\mathcal{P} = S_{2k}$, we have~\cite[Appendix A]{wiebe2010higher} $a_{\max} \leq 2k/3^k$ and $\Upsilon = 2\times 5^{k-1}$.

A \emph{symmetric} product formula is one for which $\cP(-t) = \cP^{-1}(t)$. For $p$th order formulas, the lowest $E_j$ are zero, and for symmetric formulas, the error series for $\Heff$ is even, as is captured in the following lemma.
\begin{lemma} \label{Lemma:Zero_Error_Terms}
    Let $\cP(t)$ be a $p$th order staged product formula, and suppose the BCH convergence condition of~\cref{Lemma:Effective_Hamiltonian_Error_Series} holds. Then the error operators $E_{j+1}$ from~\cref{Lemma:Effective_Hamiltonian_Error_Series} are zero for all $j < p$. Moreover, for symmetric $\cP$, $E_{j+1} = 0$ for all odd $j$.
\end{lemma}
\begin{proof}
    We will show that $\Heff(t) = H + O(t^p)$, which directly implies the first claim of the lemma. The Hamiltonians $\Heff(t)$ and $H$ may be defined through the logarithm
    \begin{equation}
        \Heff(t) = -\frac{1}{it} \log \cP (t), \qquad H = -\frac{1}{it} \log U(t)
    \end{equation}
    and, in a neighborhood of $t = 0$, the logarithm may be expanded in a power series.
    \begin{align}
    \begin{aligned}
        \Heff(t) - H &= -\frac{1}{it} \sum_{j=0}^\infty \frac{(-1)^j}{j+1} \left[(\cP(t) - \mathbbm{1})^{j+1} - (U(t) - \mathbbm{1})^{j+1}\right] \\
        &= \frac{1}{it}\sum_{j=1}^\infty \frac{(-1)^j}{j}\left[\sum_{k=0}^j (-1)^{j-k}(\cP^k(t) - U^k(t))\right]
    \end{aligned}
    \end{align}
    Since $\cP(t) - U(t) = O(t^{p+1})$ we have that $\cP(t) - U(t) = t^{p+1} E(t)$ for some analytic, operator-valued function $E$. In fact, this implies
    \begin{equation}
        \cP^k(t) - U^k(t) = t^{p+1} E^{(k)}(t)
    \end{equation}
    for some analytic $E^{(k)}$. Factoring out the $t^{p+1}$ from the series, we find that
    \begin{equation}
        \Heff(t) - H = t^p \tilde{E}(t)
    \end{equation}
    where, again, $\tilde{E}$ is some analytic operator-valued function. This shows that $\Heff(t) - H = O(t^p)$ and thus, provided the BCH series exists, all $E_j = 0$ for $j < p$. 
    
    For symmetric product formulas, the condition $\cP(-t) = \cP^\dagger(t)$ ($\cP$ is unitary) implies for the effective Hamiltonian that
    \begin{equation}
        (-i t \Heff(t))^\dagger = -i (-t) \Heff(-t)
    \end{equation}
    which further implies $\Heff(t)$ is an even function. Thus, the BCH error series is even provided it exists.
\end{proof}
\noindent As an aside, it is possible that the order conditions from~\cite[Theorem 9]{childs2021theory}, could be adapted for our purposes. However, they characterise their exponential error using a time-ordered exponential, which is suitable for their objectives but less so for ours.

We have thus characterised the error terms in $\Heff$ adequately for our purposes. Moving towards our primary interest, dynamical evolution of observables, we now wish to construct an error series for an observable evolved under a product formula. The following lemma provides what we need, and mimics the ideas of~\cite{aftab2024multi}.
\begin{lemma}\label{Lemma:Exact_Error_Form}
    Let $\cP$ be a staged $p$th order product formula and let $O$ be an observable. For any $s\in\mathbb{R}$, let
    \begin{equation*}
    \tilde{O}(T,s) \coloneqq \cP^{1/s\dagger}(sT) O \cP^{1/s}(sT)
    \end{equation*}
    be the approximate evolution of $O$ for duration $T\in\mathbb{R}$ and Trotter step size $sT$, with $s = 0$ defined via the limit. Suppose that there exists a $J\in \mathbb{Z}_+$ and $C\in\mathbb{R}_+$ such that
    \begin{equation*}
        \sup_{j\geq J} \acomm^{(j)}\left(a_\mathrm{max}\Upsilon \abs{sT}\right)^j\leq C,
    \end{equation*}
    with $a_\mathrm{max} \coloneqq \max_{\upsilon,\gamma} \abs{a_{(\upsilon,\gamma)}}$. 
    Let $\sym = 2$ if $\cP$ is symmetric, $1$ otherwise. 
    Then for any $K \in \mathbb{Z}_+$, the approximation error in $\tilde{O}(T,s)$ compared with the exact evolution $ O(T) \coloneqq e^{i H T} O e^{-i H T}$
    may be expressed as
    \begin{equation*} \label{Eq:Approx_Time_Evolution_High_Order}
        \tilde{O}(T,s) - O(T) = \sum_{j\in \sym\mathbb{Z}_+\geq p} s^j \tilde{E}_{j+1,K}(T)(O) + \tilde{F}_K(T,s)(O).
    \end{equation*}
    Here, $\tilde{E}_{j+1,K}(T)$ and $\tilde{F}_K(T,s)$ are superoperators whose induced spectral norm $\norm{\cdot}$ is bounded as 
    \begin{align*}
        \norm{\tilde{E}_{j+1,K}(T)} &\leq   (a_\mathrm{max}\Upsilon T)^j \sum_{l=1}^{\min\{K-1, \lfloor j/p \rfloor  \}}  \frac{(a_\mathrm{max}\Upsilon T)^l}{l!}\sum_{\substack{j_1\dots j_l\in \sym\mathbb{Z}_+\geq p \\ j_1+\dots+j_l=j}}  \left(\prod_{\kappa=1}^l 2 \frac{\acomm^{(j_\kappa + 1)}}{(j_\kappa + 1)^2}\right)\\
        \norm{\tilde{F}_K(T,s)} &\leq \frac{(a_\mathrm{max} \Upsilon T)^K}{K!} \sum_{j\in \sym\mathbb{Z}_+\geq Kp} (a_\mathrm{max} \Upsilon sT)^j \left( \sum_{\substack{j_1\dots j_K\in \sym\mathbb{Z}_+\geq p \\ j_1+\dots+j_K=j }} \left(\prod_{\kappa=1}^K 2 \frac{\acomm^{(j_\kappa+1)}}{(j_\kappa+1)^2}\right) \right).
    \end{align*}
\end{lemma}

\begin{proof}
    By definition, we have that 
    \begin{equation}
        \cP^{1/s}(sT) = e^{-i T \Heff(t)},
    \end{equation}
    where $t := sT$ is the Trotter step size. By~\cref{Lemma:Effective_Hamiltonian_Error_Series}, $\Heff$ is expandable as a BCH series, and we write $\Heff(t) = H + E(t)$, where
    \begin{equation}
        E(t) = \sum_{j\in \sym\mathbb{Z}_+ \geq p} E_{j+1} t^j.
    \end{equation}
    In what follows, the independent parameters are $t$ and $T$, and any $t$-dependence is left implicit. Our analysis is based on the variation of parameters formula applied to the Hamiltonian evolution of observables. For observable $O$, and Hermitian $H$ and $E$, the evolution equations for $\tilde{O}(T) := e^{i(H+E)T}Oe^{-i(H+E)T}$ and $O(T) := e^{iHT}Oe^{-iHT}$ are given by
    \begin{align}
       \partial_T \tilde{O}(T) = i[H+E,\tilde{O}(T)], \qquad
       \partial_T O(T) = i[H,O(T)].
    \end{align}
    For ease of notation, we will write this in a super-operator formalism
    \begin{align}
        \tilde{O}(T,s) = e^{i T\ad_{H+E} }O,\qquad O(T) = e^{iT\ad_{H} }O
    \end{align}
    where $\ad_X(\cdot) := [X,\cdot]$ is (anti) Hermitian with respect to the Hilbert-Schmidt inner product when $X$ is (anti) Hermitian with respect to the standard inner product. In this context, the variation of parameters formula~(\ref{eq:var_of_param_general}) gives
    \begin{equation}
        \tilde{O}(T,s) = O(T) + \int^T_0 d\tau_1 e^{i(T-\tau_1)\ad_H } i\ad_E (\tilde{O}(\tau_1)).
    \end{equation}
    Iterating this formula once,
    \begin{multline}
        \tilde{O}(T,s) = O(T) \\
        + \int^T_0 d\tau_1 e^{i(T-\tau_1)\ad_H } i\ad_E (O(\tau_1)) + \int^T_0 d\tau_1 \int_0^{\tau_1} d\tau_2 e^{i(T-\tau_1)\ad_H } i\ad_E e^{i(\tau_1-\tau_2)\ad_H } i\ad_E (\tilde{O}(\tau_2)).
    \end{multline}
    Iterating $K-1$ times gives
    \begin{align}
         \tilde{O}(T,s) &= O(T)  \nonumber
         \\
         &+ \sum_{l=1}^{K-1} \int^T_0 d\tau_1 \int^{\tau_1}_0 d\tau_2 \dots \int^{\tau_{l-1}}_0 d\tau_l e^{i(T-\tau_1)\ad_H } i\ad_E e^{i(\tau_1-\tau_2)\ad_H } i\ad_E \dots  e^{i(\tau_{l-1}-\tau_l)\ad_H } i\ad_E(O(\tau_l))   \label{Eq:Iterate_K}
         \\
         &+ \int^T_0 d\tau_1\int^{\tau_1}_0 d\tau_2 \dots \int_0^{\tau_{p-1}} d\tau_K e^{i(T-\tau_1)\ad_H } i\ad_E e^{i(\tau_1-\tau_2)\ad_H } i\ad_E \dots  e^{i(\tau_{K-1}-\tau_K)\ad_H } i\ad_E(\tilde{O}(\tau_K)). \label{Eq:Remainder_Term}
    \end{align}
    Taking line~(\ref{Eq:Iterate_K}) and expanding the definition of $E$, 
    \begin{align}
        &\sum_{l=1}^{K-1} \int^T_0 d\tau_1\int^{\tau_1}_0 d\tau_2\dots \int^{\tau_{l-1}}_0 d\tau_l e^{i(T-\tau_1)\ad_H } i\ad_E e^{i(\tau_1-\tau_2)\ad_H} i\ad_E \dots  e^{i\tau_{l-1}-\tau_l)\ad_H } i\ad_E(O(\tau_l)) 
        \nonumber\\ 
        =& \sum_{l=1}^{K-1} \int^T_0 d\tau_1 \int^{\tau_1}_0 d\tau_2  \dots \int^{\tau_{l-1}}_0 d\tau_l \left( \prod_{\kappa=l}^1\left( \sum_{j_\kappa\in \sym\mathbb{Z}_+\geq p } e^{i(\tau_{\kappa-1}-\tau_\kappa)\ad_H}    i\ad_{E_{j_\kappa+1}} t^{j_\kappa}  \right)   \right)(O(\tau_l)) 
        \\ 
        =&\sum_{l=1}^{K-1} \int^T_0 d\tau_1 \int^{\tau_1}_0 d\tau_2 \dots \int^{\tau_{l-1}}_0 d\tau_l \left( \sum_{\substack{j\in \sym\mathbb{Z}_+\geq pl}} t^j\sum_{\substack{j_1\dots j_l\in \sym\mathbb{Z}_+\geq p\\ j_1+\dots+ j_l=j}}\left(\prod_{\kappa=l}^1 e^{i(\tau_{\kappa-1}-\tau_\kappa)\ad_H}    i\ad_{E_{j_\kappa+1}}\right) \right)(O(\tau_l)) \label{Eq:E_Derivation}
    \end{align}
    with $\tau_0 \equiv T$. We now reinsert $t = sT$, and make a change of variables $s_i = \tau_i /T$. This gives
    \begin{align}
        \sum_{l=1}^{K-1} T^l \int^1_0 ds_1 \int^{s_1}_0 ds_2 \dots \int^{s_{l-1}}_0 ds_l  \left( \sum_{j\in \sym\mathbb{Z}_+ \geq pl} (sT)^j\sum_{\substack{j_1\ldots j_l\in \sym\mathbb{Z}_+\geq p \\ j_1+\dots+j_l=j}} \left(\prod_{\kappa=l}^1  e^{i(s_{\kappa-1}-s_\kappa)T\ad_H}    i\ad_{E_{j_\kappa+1}}\right)\right)(O(Ts_l)). 
    \end{align}
    Next, we regroup the sum according to the degree of $s$, which yields
    \begin{align}
       &\sum_{j \in \sym\mathbb{Z}_+\geq p} (sT)^j \sum_{l=1}^{\min\{K-1, \lfloor j/p \rfloor \}} T^l \int^1_0 ds_1 \int^{s_1}_0 ds_2 \dots \int^{s_{l-1}}_0 ds_l \sum_{\substack{j_1\dots j_l\in \sym\mathbb{Z}_+\geq p \\ j_1+\dots+j_l=j}}  \left(\prod_{\kappa=l}^1 e^{iT(s_{\kappa-1}-s_\kappa)\ad_H}    i\ad_{E_{j_\kappa+1}}\right) (O(T s_l)) 
       \nonumber \\
       = &\sum_{j \in \sym\mathbb{Z}_+\geq p}  s^j \tilde{E}_{j+1,K}(T)(O).
    \end{align}
    Here, we have defined
    \begin{multline} \label{eq:tilde_E_explicit}
        \tilde{E}_{j+1,K}(T) \coloneqq\\
        \sum_{l=1}^{\min\{K-1, \lfloor j/p \rfloor\}} T^{j+l} \int^1_0 ds_1 \int^{s_1}_0 ds_2 \dots \int^{s_{l-1}}_0 ds_l \sum_{\substack{j_1\dots j_l\in \sym\mathbb{Z}_+\geq p \\ j_1+\dots+j_l=j}}  \left(\prod_{\kappa=l}^1 e^{i(s_{\kappa-1}-s_\kappa)T\ad_H} i\ad_{E_{j_\kappa+1}}\right) e^{is_l T \ad_H}.
    \end{multline}
    We now want to put bounds on the norm of $\tilde{E}_{j+1,K}$. Using the triangle inequality, unitarity of $e^{i \tau \ad_H}$, and evaluating the remaining integral,
     \begin{align}
     \begin{aligned}
        \norm{\tilde{E}_{j+1,K}(T)} &\leq \sum_{l=1}^{\min\{K-1, \lfloor j/p \rfloor  \}} T^{j+l} \int^1_0\int^{s_1}_0 \dots \int^{s_{l-1}}_0 ds_l \sum_{\substack{j_1\dots j_l\in \sym\mathbb{Z}_+\geq p \\ j_1+\dots+j_l=j}}  \left( \prod_{\kappa=1}^l \norm{\ad_{E_{j_\kappa+1}}} \right)   \\
        &\leq  T^j \sum_{l=1}^{\min\{K-1, \lfloor j/p \rfloor  \}}  \frac{T^l}{l!}\sum_{\substack{j_1\dots j_l\in \sym\mathbb{Z}_+\geq p \\ j_1+\dots+j_l=j}}  \left(\prod_{\kappa=1}^l     \norm{\ad_{E_{j_\kappa+1}} }  \right).   
    \end{aligned}
    \end{align}
    Noting that $\norm{\ad_x} \leq 2\norm{X}$ and applying~\cref{Lemma:Effective_Hamiltonian_Error_Series},
    \begin{align}
    \begin{aligned}
        \norm{\tilde{E}_{j+1,K}(T)} &\leq   T^j \sum_{l=1}^{\min\{K-1, \lfloor j/p \rfloor  \}}  \frac{T^l}{l!}\sum_{\substack{j_1\dots j_l\in \sym\mathbb{Z}_+\geq p \\ j_1+\dots+j_l=j}}  \left(\prod_{\kappa=1}^l 2 \acomm^{(j_\kappa + 1)} \frac{(a_\mathrm{max} \Upsilon)^{j_\kappa + 1}}{(j_\kappa + 1)^2}\right) \\
        &= (a_\mathrm{max}\Upsilon T)^j \sum_{l=1}^{\min\{K-1, \lfloor j/p \rfloor  \}}  \frac{(a_\mathrm{max}\Upsilon T)^l}{l!}\sum_{\substack{j_1\dots j_l\in \sym\mathbb{Z}_+\geq p \\ j_1+\dots+j_l=j}}  \left(\prod_{\kappa=1}^l 2 \frac{\acomm^{(j_\kappa + 1)}}{(j_\kappa + 1)^2}\right).
    \end{aligned}
    \end{align} 

    So far we have considered the terms in line~(\ref{Eq:Iterate_K}).
    We now consider line~(\ref{Eq:Remainder_Term}), which we will denote as $\tilde{F}_K(T,s)$.
    Applying the triangle inequality and utilising unitarity in a similar manner as above, we can check that the operator norm of $\tilde{F}_{K}(T,s)$ is bounded by
    \begin{align}
    \begin{aligned}
        \norm{\tilde{F}_K(T,s)} &\leq \frac{T^K}{K!}\norm{\ad_E}^K
        \\
        &\leq \frac{T^K}{K!}2^K \norm{E}^K \\
        &\leq \frac{T^K}{K!}2^K\left( \sum_{j\in \sym\mathbb{Z}_+\geq p} \norm{E_{j+1}} (sT)^j\right)^K \nonumber
        \\
        &= \frac{T^K}{K!}2^K \sum_{j_1\dots j_K\in \sym\mathbb{Z}_+\geq p} \left(\prod_{\kappa=1}^K \norm{E_{j_\kappa + 1}}\right)(sT)^{j_1+\dots+j_K} 
        \\
        &\leq \frac{T^K}{K!} \sum_{j\in \sym\mathbb{Z}_+\geq Kp} (sT)^j \left( \sum_{\substack{j_1\dots j_K\in \sym\mathbb{Z}_+\geq p\\ j_1+\dots+j_K=j}} \prod_{\kappa=1}^K 2\norm{E_{j_\kappa + 1}}\right).
    \end{aligned}
    \end{align}
    Using, as before, the bounds from~\Cref{Lemma:Effective_Hamiltonian_Error_Series},
    \begin{equation}
        \norm{\tilde{F}_K(T,s)} \leq \frac{(a_\mathrm{max}\Upsilon T)^K}{K!} \sum_{j\in \sym\mathbb{Z}_+\geq Kp} (a_\mathrm{max}\Upsilon sT)^j\left( \sum_{\substack{j_1\dots j_K\in \sym\mathbb{Z}_+\geq p\\ j_1+\dots+j_K=j}} \prod_{\kappa=1}^K 2 \frac{\acomm^{(j_\kappa + 1)}}{(j_\kappa+1)^2}\right)
    \end{equation}
    giving the second bound of the lemma.
    
\end{proof}

\section{Richardson Extrapolation}
\label{sec:richardson-extrapolation}

Having laid the technical groundwork in the previous section, we now apply these results to analyse Richardson extrapolation for time evolved observables. Before we begin, we briefly provide a more detailed overview Richardson extrapolation to supplement~\cref{sec:richardson-background}. A more detailed and general discussion of the method can be found in~\cite{sidi_2003}. 

In our context $f:[0,\width] \rightarrow \mathbb{R}$ is a smooth function, such that
\begin{equation} \label{eq:Richardson_problem_f}
    f(s) = f(0) + \sum_{j=1}^n c_j s^{\sigma_j} + O(s^{\sigma_{n+1}})
\end{equation}
where the $\sigma_j\in \mathbb{Z}_+$ form an increasing sequence. 
Moreover, suppose $f(s_k)$ can be computed for a monotonically decreasing sequence of inputs $s_k \in (0,\width]$. Richardson extrapolation provides a new function $F^{(m)}(s)$ satisfying 
\begin{equation}
    F^{(m)}(s) = f(0) + \sum_{j=m}^n \tilde{c}_j s^{\sigma_j} + O(s^{\sigma_{n+1}})
\end{equation}
for some $1 \leq m \leq n$. Thus, the convergence rate to $f(0)$ for small $s$ is boosted from $O(s^{\sigma_1})$ to $O(s^{\sigma_m})$. The procedure can be applied to $C^m$ functions on $[0,c]$, and in this case the relevant expansion is the Taylor polynomial.

There are several algorithms for performing Richardson extrapolation. Regardless of how it is performed, the result is a linear combination
\begin{equation} \label{eq:Richardson_lin_sum}
    F^{(m)}(s) =\sum_{k = 1}^m b_k f(s_k)
\end{equation}
where the $s_k$ are themselves functions of $s$, and $s=\max_k s_k = s_1$. 
Moreover, $b = (b_1,\ldots, b_m)$ solves the linear system
\begin{align} \label{eq:Richard_linalg}
    V b = \hat{e}_1
\end{align}
where $\hat{e}_j$ is the $j$th standard basis vector and $V$ is a generalised $m\times m$ Vandermonde matrix, with elements
\begin{equation}
    V_{jk} = s_k^{\sigma_{j-1}}
\end{equation}
and $\sigma_0 \equiv 0$. For certain values of $\sigma_j$, an exact solution to~\cref{eq:Richard_linalg} is known. In particular, if 
\begin{align} \label{eq:basic_Vand_condition}
    \sigma_j = \eta j 
\end{align} 
for some $\eta \in \mathbb{Z}\setminus\{0\}$ and $j=0,1,\ldots,m-1$, then $V$ is the standard Vandermonde matrix in terms of $y_k \equiv x_k^\eta$. 
Thus, the inverse is known via the theory of Lagrange interpolation. 
In particular,
\begin{equation} \label{eq:exact_b}
    b_k = (V^{-1})_{k0} = \prod_{i\neq k} \frac{x_i^\eta}{x_i^\eta - x_k^\eta}.
\end{equation}
For more general $\sigma_j$, \cref{eq:Richard_linalg} still holds, but it is less clear if a closed form solution exists in the mathematics literature. This has important implications for higher order product formula simulations that we will discuss later in this section.



A proper choice of sample points is crucial for the Richardson method to be well-conditioned, hence robust to computational imprecision. 
For example, certain natural choices, such as $s_k = s/k$ for $k>0$, lead to poor conditioning. In our context, the relevant condition number is the one norm $\norm{b}_1$ of $b$, which can grow very large despite the summation constraint
\begin{align}
    \sum_k b_k = 1
\end{align}
enforced by \cref{eq:Richard_linalg}. This results in a sign problem, where small numerical errors in $f(x_k)$ are magnified by enormous $b_k$ and hence large $\norm{b}_1$, while the size of the answer $f(0)$ remains $O(1)$. 
Some choices of $s_k$, such as a geometric sequence $s_k =\omega^{k-1} s_1$ for $\omega \in (0,1)$ are provably well-conditioned. 
However, in many applications such as our own, computing $s_k$ closer to $0$ becomes intolerably expensive. To achieve low-depth Trotter evolutions, it is preferable to keep the sample points $s_k$ as far from the origin as allowable. 

Additionally, for product formula simulation, we would also like to choose $s_k = 1/r_k$ for nonzero integer $r_k$, since otherwise we would have to apply fractional Trotter steps. 
Although this is possible using quantum signal processing techniques, the overhead is potentially large and undesirable --- particularly if using a NISQ-era device.
$s_1$ is a parameter chosen appropriately to minimise the number of Trotter steps necessary while also ensuring sufficiently accurate simulations.

To meet the well-conditioning and integer query conditions we desire, we make use of results from~\cite{low2019well}. 
In particular, we let
\begin{equation}\label{Eq:Step_Size}
    r_k = r_\mathrm{scale} \left\lceil \frac{R}{\sin(\pi(2k-1)/8m)}  \right\rceil, \quad k\in \{1,\dots, m\}
\end{equation}
and we will make the explicit choice $R = \sqrt8 m/\pi$, chosen such that the $r_k$ are distinct, and $r_\mathrm{scale}\in \mathbb{Z}_+$ scales to ensure the $r_k$ are large enough to accommodate longer evolutions. This scaling has no effect on $b$ and hence the conditioning. A useful upper bound for the purpose of resource estimates is~\cite{watkins2024clock}
\begin{equation}\label{eq:Trot_step_bounds}
   m \leq r_k/r_\mathrm{scale} \leq 3 m^2 
\end{equation}

One can prove the following about this choice of extrapolation points. 
\begin{lemma}[Well-conditioned Richardson Extrapolation~\cite{low2019well}]
\label{Lemma:General_Richardson_Extrapolation}
    Let $f \in C^{2m + 2}([-1,1])$ be an even, real-valued function of $s$, and let $P_{j}$ and $R_{j}$ be the degree $j$ Taylor polynomial and Taylor remainder, respectively, such that $f(s) = P_{j}(s) + R_{j}(s)$. 
    Let 
    \begin{equation*}
        F^{(m)}(s) = \sum_{k=1}^m  b_k f(s_k)
    \end{equation*}
    be the unique Richardson extrapolation of $f$ at points $s_1, s_2, \dots s_m$ given by 
    \begin{equation*} \label{eq:well_cond_s}
        s_k = \frac{s}{r_k},
    \end{equation*}
     for $r_k$ defined in \cref{Eq:Step_Size} and $b_k$ given in \cref{eq:exact_b} for $\eta = 2$. 
     Then
     \begin{equation*} \label{eq:Richardson_cancels_errors}
         F^{(m)}(s) = f(0) + \sum_{k=1}^m b_k R_{2m}(s_k)
     \end{equation*}
     and $\norm{b}_1=O(\log m)$.
\end{lemma}
\begin{proof}
    By the extrapolation properties of $F^{(m)}$, $\sum_{k=0}^{m-1} b_k P_{2m}(x_k) = f(0)$. Equation~\eqref{eq:Richardson_cancels_errors} follows immediately. 
    The scaling of $\norm{b}_1$  with $m$ is proven in reference~\cite{low2019well} for $r_\mathrm{scale} = 1$. For $r_\mathrm{scale} \neq 0$, it is relatively straightforward to show that the same $b$ solves~\cref{eq:Richard_linalg}. Hence, $\norm{b}_1 = O(\log m)$.
\end{proof}

\subsection{Application to Time Evolved Observables}
\label{Sec:Richardson_Extrapolation_Observables}
Richardson extrapolation may be applied to the problem of computing time-evolved expectation values, where the Trotter step $s$ is the extrapolation parameter and the function of interest is
\begin{equation}
    f(s) = \expval{O_p(T, s)} \coloneqq \tr\left(\rho_0 \cP^{1/s \dagger}(sT) O \cP^{1/s}(sT)\right).
\end{equation}
Here, $p$ is the order of the product formula $\cP$. Note that $f(0)$ may be defined via the limit, and corresponds to the ideal value. 

In what follows, we will assume the extrapolation always starts from the linear or quadratic error terms, depending on the symmetry of the product formula. We do not start the extrapolation from the smallest nonzero power according to $p$. Although this may seem unnecessary for large $p$, we wish to eventually utilise the well-conditioned extrapolation of~\cref{Lemma:General_Richardson_Extrapolation}, which is begins at 2nd order. It is possible that these concerns are primarily academic, and practitioners may find it more reasonable to use an alternative scheme without the theoretical guarantees. 

The following lemma characterises the error in Richardson approach thus described, without committing yet to a specific choice of sampling points. 
\begin{lemma}[Richardson Extrapolation Error]
\label{Lemma:Richardson_Extrapolation_Error}
    Let $O$ be an observable, $H=\sum_\gamma^\Gamma H_\gamma$ be a time independent Hamiltonian, and $\rho_0$ a quantum state. Let $\cP$ be a staged $p$th order product formula of symmetry class $\sigma$, where $\sigma = 2$ if $\cP$ is symmetric, $1$ otherwise. Let
    \begin{equation*}
        \expval{O_{p,m}(T)} \coloneqq \sum_{k=1}^m b_k \expval{O_p (T,s_k)}
    \end{equation*}
    be an $m$-term Richardson extrapolation, with ascending sequence of Trotter steps $r_k=1/s_k \in \mathbb{Z}_+$, which cancel the powers $s^\sym, s^{\sym 2},\ldots, s^{\sym(m-1)}$. Suppose that, for all $k = 1,\ldots,m$, there exist $J \in \mathbb{Z}_+$ and $C \in \mathbb{R}_+$ such that 
    \begin{equation*}
        \sup_{j\geq J} (a_\mathrm{max} \Upsilon\abs{s_k T})^j \acomm^{(j)} \leq C.
    \end{equation*}
    Then, the error in the extrapolation, as compared to the exact evolution $\expval{O(T)}$, satisfies 
    \begin{equation*}
        \abs{\expval{O(T)} - \expval{O_{p,m}(T)}} \leq \norm{O} \norm{b}_1 \sum_{\substack{j\in \sym\mathbb{Z}_+\\ j\geq \sym m}}  s_1^j\left(\sum_{l=1}^{K}  \frac{(a_\mathrm{max} \Upsilon T\lambda_{j,l})^{j+l}}{l!}\right)
    \end{equation*}
    where $\norm{b}_1 = \sum_k \abs{b_k}$, $K = \left\lceil\frac{\sym m}{p}\right\rceil$, and 
    \begin{equation*}
        \lambda_{j,l} \coloneqq \bigg(\sum_{\substack{j_1\dots j_l\in \sym\mathbb{Z}_+\geq p \\ j_1+\dots+j_l=j}}  \prod_{\kappa=1}^l 2 \frac{\acomm^{(j_\kappa + 1)}}{(j_\kappa + 1)^2}\bigg)^{1/(j+l)}.
    \end{equation*}
\end{lemma}

\begin{proof}
    By~\cref{Lemma:Exact_Error_Form},
    \begin{equation}
        \expval{\tilde{O}(T,s)} = \expval{O(T)} + \sum_{j\in \sym\mathbb{Z}_+\geq p} s^j \expval{\tilde{E}_{j+1,K}(T)(O)} + \expval{\tilde{F}_K(T,s)(O)}
    \end{equation}
    and $\expval{\tilde{F}_K(T,s)} = O(s^{Kp})$. 
    The Richardson extrapolation procedure with $m$ samples will remove all terms up to $O(s^{\sym(m-1)})$ in the series.
    Choose a value of $K$ such that 
    \begin{equation} \label{eq:smallest_K}
        Kp > \sym (m-1), 
    \end{equation}
    such as $K=\left\lceil\frac{\sym m}{p} \right\rceil$.
    Then, an $m$-term Richardson extrapolation will cancel only terms in the $\tilde{E}$ series, but leave $\tilde{F}_K$ intact. 
    Thus, the Richardson extrapolation satisfies
    \begin{equation}
        \expval{O_{p,m}(T)} - \expval{O(T)} = \sum_{k=1}^m b_k \expval{R_{\sym (m-1)}(T,s_k)(O)}
    \end{equation}
    where $R_{q}$ is the Taylor remainder of degree $q$ and satisfies
    \begin{equation}\label{eq:Remainder_def}
        R_{\sym (m-1)}(T,s) \coloneqq \sum_{\substack{j\in \sym\mathbb{Z}_+\\ j\geq \sym m}} s^j \tilde{E}_{j+1,K}(T) + \tilde{F}_K(T,s).
    \end{equation}
    Applying a H{\"o}lder's inequality,
    \begin{equation}\label{eq:Holder_split}
        \abs{\expval{O_{p,m}(T)} - \expval{O(T)}} \leq \norm{b}_1 \max_k \norm{R_{\sym (m-1)}(T,s_k)(O)} \leq \norm{b}_1 \norm{O} \max_k \norm{R_{\sym (m-1)}(T,s_k)},
    \end{equation}
    and we now focus on the remainder of~\cref{eq:Remainder_def}. By the triangle inequality,
    \begin{equation} \label{eq:triangle_R}
        \norm{R_{\sym (m-1)}(T,s)} \leq \sum_{\substack{j\in \sym\mathbb{Z}_+\\ j \geq \sym m}} s^j \norm{\tilde{E}_{j+1,K}(T)} + \norm{\tilde{F}_K(T,s)}.
    \end{equation}
    For our choice of $K$, $K - 1 \leq \lfloor j/p \rfloor$ for all $j \geq  \sym m$, and we may write
    \begin{align} \label{eq:Etilde_boundK}
    \begin{aligned}
        \norm{\tilde{E}_{j+1, K}(T)} &\leq (a_\mathrm{max} \Upsilon T)^j \sum_{l = 1}^{K-1} \frac{(a_\mathrm{max} \Upsilon T)^l}{l!} \sum_{\substack{j_1\dots j_l\in \sym\mathbb{Z}_+\geq p \\ j_1+\dots+j_l=j}}  \left(\prod_{\kappa=1}^l 2 \frac{\acomm^{(j_\kappa + 1)}}{(j_\kappa + 1)^2}\right) \\
        &= \sum_{l = 1}^{K-1} \frac{(a_\mathrm{max} \Upsilon T \lambda_{j,l})^{j+l}}{l!}.
    \end{aligned}
    \end{align}
    Applying the bound on $\tilde{F}_K(T,s)$ from~\cref{Lemma:Exact_Error_Form} and the bound on $\tilde{E}_{j+1,K}$ from~\cref{eq:Etilde_boundK} into~\cref{eq:triangle_R},
    \begin{equation}
        \norm{R_{\sym (m-1)}(T,s)}\leq \sum_{\substack{j\in \sym\mathbb{Z}_+\\j\geq \sym m}} s^j \sum_{l=1}^{K-1} \frac{(a_\mathrm{max}\Upsilon T\lambda_{j,l})^{j+l}}{l!} + \sum_{\substack{j \in \sym\mathbb{Z}_+\\j\geq Kp}} s^j \frac{(a_\mathrm{max}\Upsilon T\lambda_{j,l})^{j+K}}{K!}.
    \end{equation}
    In turn, this is upper bounded by starting the $j \geq Kp$ index at $ \sym m$. 
    Rearranging, we see that there is a matching of terms and hence we can combine the sums as
    \begin{equation}
        \norm{R_{\sym (m-1)}(T,s)}\leq \sum_{\substack{j\in \sym\mathbb{Z}_+\\j\geq \sym m}} s^j \sum_{l=1}^K \frac{(a_\mathrm{max}\Upsilon T\lambda_{j,l})^{j+l}}{l!}.
    \end{equation}
    This almost amounts to the statement of the lemma. To conclude, we observe that our bound increases with larger $s$. 
    Since $s_k$ is decreasing in $k$, $\max_k s_k = s_1$. Thus,
    \begin{equation}
        \max_k \norm{R_{\sym (m-1)}(T,s_k)}\leq \sum_{\substack{j\in \sym\mathbb{Z}_+\\j\geq \sym m}} s_1^j \sum_{l=1}^K \frac{(a_\mathrm{max}\Upsilon T\lambda_{j,l})^{j+l}}{l!}.
    \end{equation}
    which, when combined with~\cref{eq:Holder_split}, gives the result of the lemma.
\end{proof}
To make use of this lemma, we would like to ensure that $\lambda_{j,l}$ adequately captures the "size" of $H$, and does not grow too large with the indices $j, l$. In particular, we will subsequently describe error bounds using the simpler parameter
\begin{equation}
    \lambda \coloneqq \sup_{\substack{j\in \sym\mathbb{Z}_+\geq \sym m \\ 1\leq l \leq K}} \lambda_{j,l}
\end{equation}
but first, we must ensure $\lambda$ exists. In fact, we have that
\begin{equation}
    \lambda_{j,l} \leq 4 \sum_{\gamma=1}^\Gamma \norm{H_\gamma}
\end{equation}
and thus $\lambda$ exists and satisfies the same bound. To prove this, we first take a triangle inequality through the expression in~\Cref{Lemma:Richardson_Extrapolation_Error}, and use the bound on $\acomm^{(j)}$ given in~\cref{eq:acomm_bound}.
\begin{equation}
    \lambda_{j,l} \leq 2 \left(\sum_\gamma \norm{H_\gamma}\right) \left(\sum_{\substack{j_1\ldots j_l \in \sym\mathbb{Z}_+\geq p\\j_1+\dots+j_l=j}} \prod_{\kappa=1}^l \frac{1}{(j_\kappa + 1)^2}\right)^{1/{(j+l)}}
\end{equation}
Next,
\begin{equation}
    \left(\sum_{\substack{j_1\ldots j_l \in \sym\mathbb{Z}_+\geq p\\j_1+\dots+j_l=j}} \prod_{\kappa=1}^l \frac{1}{(j_\kappa + 1)^2}\right)^{1/{(j+l)}} \leq \left(\sum_{\substack{j_1\ldots j_l \in \mathbb{N} \\ j_1+\dots+j_l = j}} 1\right)^{1/(j+l)} = \binom{j + l - 1}{l-1}^{1/(j+l)}.
\end{equation}
We can then use that $\binom{j + l - 1}{l-1}\leq 2^{j+l-1}<2^{j+l}$ to get that $\binom{j + l - 1}{l-1}^{1/(j+l)}\leq 2$. Put together, this provides the stated bound.

With the relevant error bounds in hand, we turn to the question of algorithmic cost to achieve an error within tolerance $\epsilon$. Given Trotter steps $(r_1, \ldots, r_m) \in \mathbb{Z}_+^m$ listed in ascending order, the maximum Trotter depth is $r_m$ and the total Trotter depth is $\sum_k r_k$. These are both important parameters for discussing the true simulation cost. Note that, in the full algorithm, the true number of Trotter steps will be higher as a certain number of repetitions are necessary for any chosen measurement protocol.

We start by using our error bounds to derive a sufficient Trotter depth to achieve precision $\epsilon$ in the estimator, assuming exactly computed expectation values throughout.
\begin{lemma}[Sufficient Trotter Depth] \label{Lemma:Sufficient_Trotter_Depth}
    Consider the Richardson extrapolation scenario described in~\cref{Lemma:Richardson_Extrapolation_Error}. Define 
    \begin{equation*}
        \lambda \coloneqq \sup_{\substack{j\in \sym\mathbb{Z}_+\geq \sym m \\ 1\leq l \leq K}} \lambda_{j,l}
    \end{equation*}
    and $s_1$ is chosen such that $ a_\mathrm{max}\Upsilon s_1 \lambda T < 1/2$.
    To achieve a relative error $\epsilon$, namely
    \begin{align*}
        \abs{\expval{O(T)} - \expval{O_{p,m}(T)}}\leq \epsilon \norm{O},
    \end{align*}
    it suffices to choose a minimum number of Trotter steps
    \begin{equation*}
         r_1 \geq (a_\mathrm{max} \Upsilon \lambda T) \left( \frac{4 \norm{b}_1}{ \epsilon}\right)^{\frac{1}{\sym m}}
    \end{equation*}
    for $a_\mathrm{max} \Upsilon \lambda T \leq 1$ ("short times") and 
    \begin{equation*}
        r_1 \geq (a_\mathrm{max} \Upsilon \lambda T)^{1 + \frac{1}{\sym m}\left\lceil \frac{\sym m}{p}\right\rceil} \left(\frac{4 \norm{b}_1}{ \epsilon} \right)^{\frac{1}{\sym m}}
    \end{equation*}
     for $a_\mathrm{max} \Upsilon \lambda T > 1$. 
\end{lemma}
\begin{proof}
    From~\cref{Lemma:Richardson_Extrapolation_Error}, 
    \begin{align}
    \begin{aligned}
        \abs{\expval{O(T)} - \expval{O_{p,m}(T)}} &\leq  \norm{O} \norm{b}_1 \sum_{\substack{j\in \sym\mathbb{Z}_+ \\ j\geq \sym m }} s_1^j \left(\sum_{l=1}^K \frac{(a_\mathrm{max}\Upsilon T\lambda_{j,l})^{j+l}}{l!}\right) \\
        &\leq \norm{O} \norm{b}_1 \sum_{\substack{j\in \sym\mathbb{Z}_+ \\ j\geq \sym m}} (s_1a_\mathrm{max} \Upsilon \lambda T)^j \sum_{l=1}^{K}   \frac{(a_\mathrm{max} \Upsilon \lambda T)^l}{l!}.
    \end{aligned}
    \end{align}
    The inner sum is a partial sum of the exponential, which we wish to upper bound with an elementary expression. Using the upper bound,
    \begin{equation}
        (a_\mathrm{max} \Upsilon \lambda T)^l \leq \max\{a_\mathrm{max} \Upsilon \lambda T, 1\}^K \equiv \eta^K
    \end{equation}
    we have 
    \begin{equation}
        \sum_{l=1}^K \frac{(a_\mathrm{max} \Upsilon \lambda T)^l}{l!} \leq \eta^K (e - 1)
    \end{equation}
    and thus,
    \begin{align}
    \begin{aligned}
        \abs{\expval{O(T)} - \expval{O_{p,m}(T)}} &\leq (e-1)\norm{O}\norm{b}_1 \eta^K \sum_{\substack{j\in \sym\mathbb{Z}_+ \\ j \geq \sym m}} (s_1 a_\mathrm{max} \Upsilon \lambda T)^j \\
        &= (e-1) \norm{O} \norm{b}_1 \eta^K (s_1a_\mathrm{max} \Upsilon \lambda T)^{\sym m} \sum_{\substack{j\in \sym\mathbb{Z}_+\\j\geq0}} (s_1 a_\mathrm{max} \Upsilon \lambda T)^j \\
        &\leq (e-1) \norm{O} \norm{b}_1 \eta^K (s_1a_\mathrm{max} \Upsilon \lambda T)^{\sym m} \sum_{\substack{j\in \sym\mathbb{Z}_+\\j\geq0}} \left(\frac{1}{2}\right)^j\\
        &\leq 4 \norm{O}\norm{b}_1 \eta^K (s_1 a_\mathrm{max} \Upsilon \lambda T)^{\sym m}.
    \end{aligned}
    \end{align}
    To achieve an relative error (i.e., normalised by $\norm{O}$) of $\epsilon$, it suffices then to choose $r_1=1/s_1$ satisfying
    \begin{equation}
        r_1 \geq a_\mathrm{max} \Upsilon \lambda T \left(\frac{4 \norm{b}_1 \eta^K}{\epsilon}\right)^{\frac{1}{\sym m}},
    \end{equation}
    and we may simply take the ceiling of the right hand side as our value. 

    We now split into the short and long-time regimes. For short times $a_\mathrm{max} \Upsilon \lambda T \leq 1$, $\eta = 1$ and we have
    \begin{equation}
        r_1 = \left\lceil a_\mathrm{max} \Upsilon \lambda T \left(\frac{4\norm{b}_1}{\epsilon}\right)^{\frac{1}{\sym m}}\right\rceil.
    \end{equation}
    On the other hand, for long times $(a_\mathrm{max} \Upsilon \lambda T > 1)$, we choose
    \begin{equation}
        r_1 = \left\lceil (a_\mathrm{max} \Upsilon \lambda T)^{1 + \frac{1}{\sym m}\left\lceil \frac{\sym m}{p}\right\rceil}\left(\frac{4\norm{b}_1}{\epsilon}\right)^{\frac{1}{\sym m}}\right\rceil.
    \end{equation}
This gives the lemma statement.
\end{proof}

\noindent So far we have derived a set of bounds for the minimum number of Trotter steps required to reach a given error.
We now examine the asymptotic scaling of the parameters.

\begin{corollary}[Asymptotic Trotter Costs]\label{Corollary:Max_Trotter_Steps}
    Consider an $m$-term Richardson extrapolation of a time-evolved expectation value in the setting of the previous lemma, where $\cP$ is symmetric.
    Choose the Trotter step size according to~\cref{Lemma:General_Richardson_Extrapolation}, with $r_\mathrm{scale}$ large enough such that $r_1$ satisfies the "long time" condition of~\cref{Lemma:Sufficient_Trotter_Depth} for a choice of $m$ scaling as $O(\log(1/\epsilon))$.
    Then, the maximum number of Trotter steps that needs to be implemented scales as
    \begin{equation*}
        \max_k r_k = O\left((a_\mathrm{max} \Upsilon \lambda T)^{(1+1/p)}\log(1/\epsilon)\right).
    \end{equation*}
\end{corollary}
\begin{proof}
By satisfying the conditions of~\cref{Lemma:Sufficient_Trotter_Depth}, the extrapolation scheme achieves a relative error $\epsilon$, in the sense stated in that lemma. Since $\norm{b}_1=O(\log m)$, $\norm{b}_1^{1/2m}=O(1)$. Choose $m= p \left\lceil\log\left( \frac{1}{\epsilon}\right) \right\rceil $. Then,
\begin{align}
    \left(\frac{4}{\epsilon} \right)^{\frac{1}{2p\lceil \log(1/\epsilon)\rceil}} = O(1).
\end{align}
Moreover, since $m$ is a multiple of $p$, $\frac{1}{\sym m}\lceil \frac{\sym m}{p} \rceil=\frac{1}{p}$.
Putting these into~\cref{Lemma:Sufficient_Trotter_Depth},
\begin{align}
    r_1 = O\left((a_\mathrm{max} \Upsilon \lambda T)^{(1+1/p)}\right)
\end{align}
is the minimum number of Trotter steps. To obtain an upper bound on the maximum number of steps, we utilise the bounds~\cref{eq:Trot_step_bounds} to obtain
\begin{align}
\begin{aligned}
    r_k &\leq 3 m r_1 \\
    &= O\left((a_\mathrm{max} \Upsilon \lambda T)^{(1+1/p)} m\right)
\end{aligned}
\end{align}
given our choice of $m$, this yields the scaling stated in the corollary.
\end{proof}
\noindent We remark that, in the above proof, the choice to make $m$ a multiple of $p$ is mainly for simplicity of proof. For our purposes, we treat $p$ as fixed and not scaling with the simulation parameters.

\subsection{Resource Estimates for Richardson Extrapolation}
\label{Sec:Resource_Estimates_Richardson}

So far, we have determined error bounds for the Richardson procedure and provided partial results on the resources required. 
In this section, we will derive full resource costs. We examine two metrics: the maximum circuit depth $D_\mathrm{max}$ of Trotter steps needed, and the total number $\Cexp$ of Trotter steps required.
The former is arguably the most relevant metric for NISQ-era devices where only short-depth circuits are possible to implement, whereas $\Cexp$ is more relevant for fault tolerant devices.
We note that $\Cexp$ is proportional to the total number of elementary operations required for the protocol.

There are two primary sources of error we consider: the extrapolation error (associated with the Richardson extrapolation procedure) and the error associated with the measurement protocol (e.g., shot noise). Notably, we fully neglect "physical" errors such as gate imperfections or decoherence.
We suppose the estimates for the function $f$ at points $\{s_1\}_{k=1}^m$ are given by $\tilde{f}(s_k)$, such that the final estimate we have is $\tF^{(m)}(s) = \sum_{k=1}^m b_k \tilde{f}(s_k)$.
The error in our final prediction is $\epsilon = |f_B(0) - \tF^{(m)}(s)|$ which can be broken down as
\begin{align}
    |f_B(0) - \tilde{F}^{(m)}(s)| &\leq |f_B(0) - F^{(m)}(s)| + |F^{(m)}(s) - \tilde{F}^{(m)}(s)|  \nonumber \\
     &\leq |f_B(0) - F^{(m)}(s)| + \left|\sum_{k=0}^{m-1} b_k f(s_k) - \sum_{k=0}^{m-1} b_k \tilde{f}(s_k) \right|  \nonumber \\
    &\leq \norm{O}(\epsint + \norm{b}_1\epsdata), \label{Eq:Error_Decomposition}
\end{align}
where  $\epsint$ is the (relative) interpolation error and $ |f(s_k) - \tilde{f}(s_k)| \leq \norm{O}\epsdata$ is the maximum (relative) error associated with each individual measurement point.
To satisfy a total relative error tolerance $\epsilon$, it thus suffices to ensure that
 \begin{equation} \label{eq:Error_Conditions}
    \epsint \leq \frac{\epsilon}{2}, \quad \quad \epsdata \leq \frac{\epsilon}{2\norm{b}_1}.
\end{equation}
The resources required to satisfy the first of these inequalities is essentially the content of~\cref{Corollary:Max_Trotter_Steps}, because the factor of $1/2$ will not affect the asymptotics. In the following subsections, we look more closely at the resources needed to have sufficiently small error in the data, then obtain an overall cost bound.

\subsubsection{Incoherent Measurements} \label{Sec:Resource_Costs_Incoherent}

Within the incoherent scheme, let's consider how many measurements are needed to ensure $\epsdata \leq \epsilon/(2\norm{b}_1)$. From Hoeffding's inequality, we see that to achieve an estimate $\expval{O}'$ satisfying 
\begin{equation}
    |\tr[\cP^{\dagger 1/s}(sT)O\cP^{1/s}(sT)\rho_0]- \expval{O}'|\leq \epsdata \norm{O}
\end{equation}
with probability $\geq (1-\delta')$, it suffices to choose a number of samples $N$ satisfying
\begin{align} \label{eq:Hoeffding_incoherent}
    N \geq \frac{1}{2\epsdata^2}\log\left(\frac{2}{\delta'}\right).
\end{align}
By the union bound, it suffices to choose $\delta' = \delta/m$ to have an overall success probability of $1-\delta$ for all measurements. Taking $m = O(\log(1/\epsilon))$ from~\cref{Corollary:Max_Trotter_Steps}, we have that $N$ scales as
\begin{equation} \label{eq:N_scaling}
    N = O\left(\frac{\norm{b}_1^2}{\epsilon^2} \log\left(\frac{m}{\delta}\right)\right) = O\left(\frac{(\log m)^3}{\epsilon^2 }\right) = O\left(\frac{(\log\log 1/\epsilon)^3}{\epsilon^2}\right).
\end{equation}

Each measurement requires one product formula evolution. The maximum Trotter step size is given in \cref{Corollary:Max_Trotter_Steps}, and this directly gives the maximum Trotter depth. Meanwhile, the total number of Trotter steps required (which is proportional to the total resources) scales as
\begin{align}
\begin{aligned}
     \Cexp &\leq N  \sum_{k=1}^m r_k  \\
     &\leq N r_1 \sum_{k=1}^m (r_k/r_1) \\
     &\leq N r_1 \frac{1}{m}\sum_{k=1}^m \left\lceil \frac{R}{\sin(\pi(2k+1)/8m)}  \right\rceil
\end{aligned}
\end{align}
where we have used that $1/r_1 \leq 1/(r_\mathrm{scale} m)$ from~\cref{eq:Trot_step_bounds}. Next, we have
\begin{align}
\begin{aligned}
    &\leq N r_1 \frac{1}{m}\sum_{k=1}^m  \left( \frac{R}{\sin\left(\frac{\pi(2k+1)}{8m}\right)} +1\right) \label{Eq:Sum_Bound} \\
    &\leq N r_1  \left( \sum_{k=1}^{m}  \frac{8 R}{\pi(2k+1)} + 1\right)  \\
    &\leq  N r_1 \left(R \frac{4}{\pi} (2+ \log m) + 1\right) \\
    &= O(N(a_\mathrm{max} \Upsilon \lambda T)^{1+1/p} m \log(m)).
\end{aligned}
\end{align}
Using the scaling from~\cref{eq:N_scaling}, and $m = O(\log(1/\epsilon))$, we arrive at our main result concerning the use of Richardson extrapolation with incoherent measurements.
\begin{theorem}[Resource Costs for Incoherent Measurements]
Let $\expval{O_{p,m}(T)}$ be the $m$-term Richardson extrapolation estimate for $\expval{O(T)}$, taken by varying the step size of a $p^{th}$-order staged product formula, with samples taken at the rescaled Chebyshev nodes as specified in \cref{Eq:Step_Size}. 
The resource costs for computing this estimate such that 
\begin{align*}
    |\expval{O(T)} - \expval{O_{p,m}(T)}|\leq \epsilon\norm{O},
\end{align*}
when using incoherent measurements uses a number of sample points $m = O(1/\epsilon)$. 
Moreover, the maximum Trotter depth and total Trotter steps scales as
\begin{equation*}
    D_\mathrm{max} = O\left( (a_\mathrm{max} \Upsilon \lambda T)^{1+1/p}\log( 1/\epsilon) \right),\qquad C_\mathrm{Trot} = O\left( \frac{(a_\mathrm{max} \Upsilon \lambda T)^{1+1/p}}{\epsilon^2}\log(1/\epsilon) (\log\log(1/\epsilon))^4\right).
\end{equation*}
\noindent Here $\lambda\leq 4 \sum_\gamma \norm{H_\gamma}$ is defined in \cref{Lemma:Sufficient_Trotter_Depth} and there is a failure probability of $\delta = 0.01$.
\end{theorem}

\subsubsection{Coherent Measurements} \label{Sec:Resource_Costs_Coherent}

Although incoherent measurements are conceptually simple and easy to implement, they give $O(1/\epsilon^2)$ scaling rather than the optimal Heisenberg scaling of $O(1/\epsilon)$.
In this section, we consider the Iterative Quantum Amplitude Estimation (IQAE) scheme developed by~\citeauthor{grinko2021iterative} for this problem. Although the method deviates by $O(\log\log 1/\epsilon)$ from the ideal Heisenberg scaling, it achieves better constant factors among rigorously-analysed methods with true Heisenberg scaling that do not require standard Quantum Phase Estimation (QPE), a relatively intensive routine with large qubit overhead~\cite{grinko2021iterative}. Thus, the method represents a practical yet rigorous algorithm suitable for our purposes. We follow a similar analysis to~\cite{rendon2024improved}. 

We assume the problem of measuring the expectation value can be written as the problem of estimating an amplitude, as is often done in practice via, say, the Hadamard test.
Assuming $\norm{O}\leq 1$ (if otherwise, we can rescale it), then performing a Hadamard test gives an amplitude $\frac{1+\expval{O_p(T,s)}}{2}$.
This amplitude can now be estimated using IQAE.
In particular, for each sample point $s_k$, we need a number of calls to a Grover oracle $N_G$ scaling as \cite[Section E]{rendon2024improved}
\begin{align}
    N_G \leq \frac{100}{\epsdata}\log\left(\frac{2m}{\delta}\log\left(\frac{\pi}{\epsdata}\right) \right)
\end{align}
where $\delta/m$ is the probability of failure per data point, which by the union bound ensures an overall success rate of $1-\delta$. For a particular value of Trotter step $s$, we need a total Trotter depth scaling as $N_G/s$. The maximum Trotter step size needed to compute $F^{(m)}(s)$ is $\max_k r_k$, which is characterised asymptotically in~\cref{Corollary:Max_Trotter_Steps}. Thus, the maximum Trotter depth is
\begin{align*}
    \Dmax &\leq N_G \max_k r_k \\
    &=  O\left((a_\mathrm{max} \Upsilon \lambda T)^{1+1/p} \log(1/\epsilon) \frac{\norm{b}_1}{\epsilon} \log\left( \frac{\log(1/\epsilon)}{\delta} \log \bigg( \frac{\norm{b}_1}{\epsilon} \bigg) \right)\right) \\
    &= O\left( \frac{(a_\mathrm{max} \Upsilon \lambda T)^{1+1/p}}{\epsilon}\log(1/\epsilon) (\log\log(1/\epsilon))^2\right)
\end{align*}
Meanwhile, the total number of Trotter steps, $\Cexp$, is bounded as
\begin{equation}
    \Cexp \leq N_G \sum_{k=1}^m r_k.
\end{equation}
Using the same reasoning as in~\cref{Eq:Sum_Bound}, with $N$ replaced by $N_G$, we obtain
\begin{align}
\begin{aligned}
    \Cexp &= O\left(N_G (a_\mathrm{max} \Upsilon \lambda T)^{1+1/p} m \log m\right)
    \\
    &=O\left(\frac{(a_\mathrm{max} \Upsilon \lambda T)^{1+1/p}}{\epsilon} \log(1/\epsilon) (\log\log 
 (1/\epsilon))^3\right).
\end{aligned}
\end{align}
We summarise the results of the above analysis in our main result for Richardson extrapolation using coherent measurements.
\begin{theorem}[Resource Costs for Coherent Measurements]
Let $\expval{O_{p,m}(T)}$ be the $m$-term Richardson extrapolation estimate for $\expval{O(T)}$, taken by varying the step size of a $p^{th}$-order symmetric staged product formula, with samples taken at the rescaled Chebyshev nodes as specified in \cref{Eq:Step_Size}. 
The resource costs for computing this estimate such that 
\begin{align*}
    |\expval{O(T)} - \expval{O_{p,m}(T)}|\leq \epsilon\norm{O},
\end{align*}
when using coherent measurements is given by
\begin{equation*}
    D_\mathrm{max} = O\left( \frac{(a_\mathrm{max} \Upsilon \lambda T)^{1+1/p}}{\epsilon}\log( 1/\epsilon) (\log\log(1/\epsilon))^2\right), \qquad C_\mathrm{Trot} = O\left( \frac{(a_\mathrm{max} \Upsilon \lambda T)^{1+1/p}}{\epsilon}\log(1/\epsilon)(\log\log(1/\epsilon))^3\right)
\end{equation*}
\noindent where $\lambda \leq \sum_\gamma \norm{H_\gamma}$ is defined in \cref{Lemma:Sufficient_Trotter_Depth} and there is a failure probability of $\delta = 0.01$. Moreover, $m = O(\log(1/\epsilon))$.
\end{theorem}

\section{Polynomial Interpolation}
\label{sec:polynomial-interpolation}

We now consider the polynomial interpolation algorithm of~\citeauthor{rendon2024improved} for mitigating Trotter errors. Let $f(s)\in C^m[-\width,\width]$ be real-valued function, and suppose we have the value of $f$ at points $s_1, s_2, \dots s_m$. 
Let $P_{m-1} f$ be the unique degree $m-1$ polynomial interpolating $f$ at the $s_k$. It is possible to show~\cite{quarteroni2010numerical} that the error of the approximating polynomial in the interval $[-\width,\width]$ is bounded as
\begin{align*}
    |f(s) - P_{m-1} f(s)| \leq \max_{\xi \in [-\width,\width]} \frac{|f^{(m)}(\xi)|}{m!} |\omega_m(s)|,
\end{align*}
where 
\begin{align*}
    \omega_m(s) \coloneqq \prod^m_{k=1}(s-s_k) 
\end{align*}
is the monic nodal polynomial. If we choose our samples to be taken at the Chebyshev nodes on $[-\width,\width]$, given by
\begin{align}\label{Eq:Standard_Chebyshev_Nodes}
    s_i = \width \cos \left( \frac{2i-1}{2m}\pi \right), \quad \quad i\in \{1,2,\dots , m\},
\end{align}
then the interpolation satisfies a number of nice properties, such as robustness to errors in the interpolation values.
As in the Richardson extrapolation case, we are interesting in bounding the error at $s=0$.
To do so, we can use the following lemma.
\begin{lemma}[Lemma 2, \cite{rendon2024improved}]\label{Lemma:Polynomial_Approx_Error}
    Let $s_1, s_2, \ldots, s_m$ be the collection of Chebyshev interpolation points on the interval $[-\width,\width]$.
    Then the error of the interpolating polynomial $P_{m-1} f$ with respect to $f$ at $s=0$ is bounded as
    \begin{equation*}
        \abs{f(0) - P_{m-1}f(0)} \leq \max_{s \in [-\width,\width]} \abs{f^{(m)}(s)} \left(\frac{\width}{2m}\right)^m.
\end{equation*}
\end{lemma}

\subsection{Application to Time Evolved Observables}
\label{Sec:Sec:Polynomial_Interpolation_Observables}
From \cref{Lemma:Polynomial_Approx_Error}, we see that the key property in bounding the error to polynomial interpolation is to bound the derivatives as a function of $s$.
We start by bounding the derivatives at the origin.

\begin{lemma} \label{Lemma:Derivative_Bound}
    Under the assumptions and notation of~\cref{Lemma:Exact_Error_Form}, consider an observable $O$ time evolved for time $T$ under a $p^{th}$-order staged product formulae with step size $sT$, denoted $\tilde{O}(T,s)$. Then $\tilde{O}(T,s)$ is analytic in $s$ within a neighbourhood of the origin, and the derivatives are given by
    \begin{align*}
        \partial_s^j \tilde{O}(T,0) = j! \tilde{E}_{j+1,K}(T)(O)
    \end{align*}
    for any choice of $K>\lceil j/p\rceil$. In particular, for any such $K$, $\tilde{E}_{j+1} \equiv \tilde{E}_{j+1, K}$ is the $j$th coefficient in the Taylor series for $\tilde{O}(T,s)$ at $s = 0$, and is independent of $K$.
\end{lemma}
\begin{proof}
    First, $\tilde{O}(T,s)$ is an analytic function of $s$ in a neighbourhood of the origin, being an exponential of the effective Hamiltonian, which is analytic by~\cref{Lemma:Effective_Hamiltonian_Error_Series}. For convenience, we reproduce the primary equation of~\cref{Lemma:Exact_Error_Form}.
    \begin{equation} 
        \tilde{O}(T,s) = O(T) + \sum_{j\in \mathbb{Z}_+\geq p} s^j \tilde{E}_{j+1,K}(T)(O) + \tilde{F}_K(T,s)(O)
    \end{equation}
    The term $\tilde{F}_p(T,s)(O)$ is of order at least $s^{pK}$.
    Thus, choosing $K>\lceil j/p \rceil$ means that the Taylor coefficient of order $s^j$ is simply $\tilde{E}_{j+1,K}(T)(O)$. Since the Taylor coefficient of $\tilde{O}$ is $K$-independent, the $\tilde{E}_{j+1,K}$ are in fact $K$-independent for all such choice of $K$.
\end{proof}
\noindent We emphasise in passing that the independence of $\tilde{E}_{j+1,K}$ on $K$ for $K > \lceil j/p\rceil$ can be verified from the explicit expression for $\tilde{E}_{j+1,K}$ in~\cref{eq:tilde_E_explicit}. 

From the bounds previously derived in~\cref{Lemma:Exact_Error_Form} we can straightforwardly obtain a bound on the derivatives in terms of fundamental simulation parameters.
\begin{lemma}\label{Lemma:Derivative_Bounds}
    In the notation of~\cref{Lemma:Derivative_Bound}, for any $j \geq p$, we have
    \begin{equation*}
        \norm{\tilde{E}_{j+1}(T)} \leq 2 (a_\mathrm{max }\Upsilon \lambda T)^{j+1}.
    \end{equation*}
    for "short times" $\amax\Upsilon\lambda T \leq 1$ and 
    \begin{equation*}
        \norm{\tilde{E}_{j+1}(T)} \leq 2  (a_\mathrm{max }\Upsilon \lambda T)^{j(1+1/p)}
    \end{equation*}
    for "long times" $\amax\Upsilon\lambda T > 1$.
\end{lemma}
\begin{proof}
    For convenience we reproduce the bound on $\norm{\tilde{E}_{j+1,K}(T)}$ from \cref{Eq:Approx_Time_Evolution_High_Order}.
    \begin{equation}
        \norm{\tilde{E}_{j+1,K}(T)} \leq   (a_\mathrm{max}\Upsilon T)^j \sum_{l=1}^{\min\{K-1, \lfloor j/p \rfloor  \}}  \frac{(a_\mathrm{max}\Upsilon T)^l}{l!}\sum_{\substack{j_1\dots j_l\in \sym\mathbb{Z}_+\geq p \\ j_1+\dots+j_l=j}}  \left(\prod_{\kappa=1}^l 2 \frac{\acomm^{(j_\kappa + 1)}}{(j_\kappa + 1)^2}\right)
    \end{equation}
    Choosing $K>\lceil j/p \rceil$ and invoking the definition of $\lambda$,
    \begin{align}
          \norm{\tilde{E}_{j+1,K}(T)} &\leq   (a_\mathrm{max}\Upsilon T)^j \sum_{l=1}^{\lfloor j/p \rfloor}  \frac{(a_\mathrm{max}\Upsilon T)^l}{l!}\sum_{\substack{j_1\dots j_l\in \sym\mathbb{Z}_+\geq p \\ j_1+\dots+j_l=j}}  \left(\prod_{\kappa=1}^l 2 \frac{\acomm^{(j_\kappa + 1)}}{(j_\kappa + 1)^2}\right)
         \\
         &\leq 
         \sum_{l=1}^{\lfloor j/p \rfloor}  \frac{(a_\mathrm{max}\Upsilon T \lambda )^{j+l}}{l!}.
    \end{align}
    First, consider the case $\amax\Upsilon\lambda T \leq 1$. We have
    \begin{align}
    \begin{aligned}
        \norm{\tilde{E}_{j+1,K}(T)} &\leq (\amax\Upsilon\lambda T)^{j+1} \sum_{l=1}^{\lfloor j/p\rfloor} \frac{1}{l!} \\
        &< (\amax\Upsilon\lambda T)^{j+1} (e-1) \\
        &< 2 (\amax\Upsilon\lambda T)^{j+1}.
    \end{aligned}
    \end{align}
    This provides the first bound of the corollary. For the long-time case, we instead use the bound
    \begin{equation}
        \norm{\tilde{E}_{j+1,K}(T)} \leq (\amax\Upsilon\lambda T)^{j + \lfloor j/p \rfloor}  \sum_{l=1}^{\lfloor j/p \rfloor} \frac{1}{l!} < 2(\amax\Upsilon\lambda T)^{j(1+1/p)}.
    \end{equation}
    These two bounds gives the stated result.
\end{proof}
\noindent As an aside, for $0 < j < p$, the $\tilde{E}_{j+1}$ simply vanish because of the properties of the error series for order $p$ formulae. Thus, we simply exclude this case from the above corollary.

\begin{corollary}[Taylor Series for Trotter-Evolved Observable]
\label{Corollary:Taylor_Series}
    The Trotter-evolved observable can be expressed as a Taylor series
     \begin{align*}
        \tilde{O}(T,s) = O(T) +\sum_{j\in \sym\mathbb{Z} \geq p} s^j \tilde{E}_{j+1}(T)(O),
    \end{align*}
    where $\norm{\tilde{E}_{j+1}}$ is bounded as per \cref{Lemma:Derivative_Bounds}.
\end{corollary}
\begin{proof}
    Follows by constructing the Taylor series from the derivatives using \cref{Lemma:Derivative_Bound}.
    The bounds on the derivatives are then given in \cref{Lemma:Derivative_Bounds}.
\end{proof}

We can now apply this to bound the error of the polynomial interpolation procedure.
\begin{theorem}[Polynomial Interpolation Error]
\label{Theorem:Polynomial_Interpolation_Error}
    Consider a Chebyshev interpolation $P_{m-1} f(s)$ of the time evolved expectation value $f(s)$ on the interval $[-\width,\width]$, for long simulation time $\amax\Upsilon\lambda T > 1$, with
    \begin{equation*}
        \width = \frac{1}{2}(\amax \Upsilon\lambda T)^{-(1+1/p)}.
    \end{equation*}
    The approximation error at $s=0$ may be bounded as
    \begin{equation*}
        \frac{\abs{f(0) - P_{m-1}f(0)}}{\norm{O}} \leq c e^{-\gamma m}
    \end{equation*}
    where $\gamma > 1.5$ and $c < 11$.
\end{theorem}
\begin{proof}
    From~\cref{Corollary:Taylor_Series}, the Taylor series for $\tilde{O}(T,s)$ is given by
    \begin{align*}
        \tilde{O}(T,s) = O(T) +\sum_{j\in \sym\mathbb{Z} \geq p} s^j \tilde{E}_{j+1}(T)(O).
    \end{align*}
    The $m$th derivative of $\tilde{O}$ is then given by
    \begin{align}
    \begin{aligned}
        \partial_s^m \tilde{O}(T,s) &= \sum_{\substack{j \in \sym\mathbb{Z}\\ j\geq \max\{p,m\}}} \frac{j!}{(j-m)!} s^{j-m} \tilde{E}_{j+1}(T)(O)\\
        &= m! \sum_{\substack{j \in \sym\mathbb{Z}\\ j\geq \max\{p,m\}}} \binom{j}{m} \tilde{E}_{j+1}(T)(O)  s^{j-m}.
    \end{aligned}
    \end{align}
    Applying the triangle inequality and noting that $\abs{s}\leq \width$,
    \begin{equation}
        \max_{s\in[-\width,\width]} \norm{\partial_s^m \tilde{O}(T,s)} \leq m! \norm{O}\sum_{\substack{j \in \sym\mathbb{Z}\\ j\geq \max\{p,m\}}} \binom{j}{m} \width^{j-m} \norm{\tilde{E}_{j+1}(T)}.
    \end{equation}
    Using~\cref{Lemma:Derivative_Bounds} in the long-time regime,
    \begin{equation}
       \max_{s\in[-\width,\width]} \norm{\partial_s^m \tilde{O}(T,s)} \leq 2\frac{m!}{\width^m}\norm{O} \sum_{\substack{j \in \sym\mathbb{Z}\\ j\geq \max\{p,m\}}} \binom{j}{m} \left(\width(\amax\Upsilon\lambda T)^{1+1/p}\right)^j.
    \end{equation}
    We now choose $\width = \frac{1}{2}(\amax \Upsilon\lambda T)^{-(1+1/p)}$. Doing so, the infinite series is given by
    \begin{equation}
        \sum_{\substack{j \in \sym\mathbb{Z}\\ j\geq \max\{p,m\}}} \binom{j}{m} \frac{1}{2^j}
    \end{equation}
    and may be upper bounded by
        \begin{equation}
        \sum_{j = m}^\infty \binom{j}{m} \frac{1}{2^j} = 2.
    \end{equation}
    Thus, the derivative is upper bounded as
    \begin{equation}
       \max_{s\in[-\width,\width]} \norm{\partial_s^m \tilde{O}(T,s)} \leq 4 \norm{O} \frac{m!}{\width^m}.
    \end{equation}
    Applying \cref{Lemma:Polynomial_Approx_Error} gives
    \begin{equation}
        \abs{f(0) - P_{m-1}f(0)} \leq 4 \norm{O} \frac{m!}{(2m)^m}.
    \end{equation}
    Using a simplified Stirling-type upper bound $m! < \sqrt{2\pi m} (m/e)^m e^{1/(12 m)}$, one obtains
    \begin{align}
    \begin{aligned}
        \abs{f(0) - P_{m-1}f(0)} &\leq 4 \sqrt{2\pi} e^{1/12} \sqrt{m} (2e)^{-m} \norm{O} \\
        &< c e^{-\gamma m} \norm{O}
    \end{aligned}
    \end{align}
    where $\gamma \equiv \ln 2 + 1 - 1/(2 e) \approx 1.509$ and $c \equiv 4 \sqrt{2\pi} e^{1/12} \approx 10.9$. This immediately leads to the statement of the lemma.
\end{proof}

\noindent It is interesting to note that, while the factor of $1/2$ in $\width$ was chosen for simplicity, other choices will lead to different values of $c, \gamma$. In particular, a factor approaching $1$ should make the exponential decay more shallow and increase the constant factor. However, it is interesting that the form of the error dependence remains exponentially decaying.

We now have the ability to estimate the number of Trotter steps necessary to achieve a certain accuracy in the interpolation, assuming the values $f(s_k)$ are computed exactly.
\begin{lemma}
    In the setting of~\cref{Theorem:Polynomial_Interpolation_Error}, to achieve a relative error $\epsilon$ such that
    \begin{align*}
       |O(T) -  P_{m-1}f(0)| \leq \epsilon\norm{O},
    \end{align*}
    it suffices to choose $m=O(\log(1/\epsilon))$ and $\width=\frac{1}{2}(\amax \Upsilon\lambda T)^{-(1+1/p)}$. This interpolation protocol requires a maximum number of Trotter steps
    \begin{align*}
        \max_k r_k = O\left((\amax \Upsilon\lambda T)^{(1+1/p)}\log(1/\epsilon)\right). 
    \end{align*}
\end{lemma}
\begin{proof}
    From~\cref{Theorem:Polynomial_Interpolation_Error}, to achieve relative error $\epsilon$, it suffices to choose $m$ such that 
    \begin{equation}
        c e^{-\gamma m} \leq \epsilon 
    \end{equation}
    i.e., $m = O(\log(1/\epsilon))$. To find the maximum number of Trotter steps needed for an particular sampling point, we see from \cref{Eq:Standard_Chebyshev_Nodes} that the smallest $s_k$ is given by $k = m/2$ and is of size $O(m^{-1})$. Hence,
    \begin{equation}
        \max_k r_k = O(\width m) = O((\amax \Upsilon\lambda T)^{(1+1/p)}\log(1/\epsilon)).
    \end{equation}
\end{proof}

\subsection{Resource Estimates}
\label{Sec:Resource_Estimates_Polynomial}
As for the Richardson extrapolation case, for polynomial interpolation we will have to deal with both the algorithmic extrapolation error $\epsint$ and the measurement error $\epsdata$.
We assume that each $f(s_i)$ is measured with error $\leq \epsdata$, i.e. we measure a set of data points $\{f_i\}_i$ such that $|f(s_i) - \tilde{f}_i|\leq \epsdata$.
Let $\tilde{P}_mf$ be the polynomial obtained from fitting to the data $\{\tilde{f}_i\}_i$, then it can be shown \cite{rivlin2020chebyshev} that
\begin{align*}
    \max_{s\in[-\width,\width]}|P_mf(s) - \tilde{P}_mf(s)|\leq L_m\epsdata,
\end{align*}
where $L_m$ is the Lebesgue constant, which for Chebyshev interpolation is bounded by $\frac{2}{\pi}\log(m+1)+1$.

Thus, in terms of the total error, $\epsilon$, we can make the partition
 \begin{align*}
    \epsint =\frac{\epsilon}{2}, \qquad \epsdata = \frac{\epsilon}{2L_m }.
\end{align*}

\noindent To find the total resource costs for the polynomial interpolation methods, we realise that the analysis is identical to the analysis performed in \cref{Sec:Resource_Costs_Incoherent} and \cref{Sec:Resource_Costs_Coherent}.

\subsubsection{Stability to Imperfect Chebyshev Nodes}
As in the Richardson extrapolation case, when performing polynomial interpolation, we wish to ensure our samples are taken from inverse integer values which are at or close to the ideal Chebyshev nodes. 
However, in \cref{sec:Imperfect_Chebyshev_Nodes}, we see that by choosing a slightly larger value of $\width$, we get robustness to this sampling error but with a new Lesbegue constant $L_m'$ which satisfies $L_m'\leq 2L_m$.
Thus, rather than sampling from the Chebyshev nodes exactly, we can instead sample from the closest inverse integer.
The robustness condition is satisfied by choosing
\begin{align}
\begin{aligned}
     1/\width  &=   O\left( (a_\mathrm{max} \Upsilon \lambda T)^{1+1/p} m^2\log\left(m \right) \right)
    \\  
    &= O\left( (a_\mathrm{max} \Upsilon \lambda T)^{1+1/p} \log^2\left(\frac{1}{\epsilon} \right)\log\log\left( \frac{1}{\epsilon}\right) \right).
\end{aligned}
\end{align}
Thus the maximum number of Trotter steps scales as
\begin{align*}
    O\left( (a_\mathrm{max} \Upsilon \lambda T)^{1+1/p} \log^3\left(\frac{1}{\epsilon} \right)\log\log\left( \frac{1}{\epsilon}\right) \right).
\end{align*}

\noindent  
Considering the total number $\Cexp$ of Trotter steps only changes the overall scaling by $\log\log$ factors in $1/\epsilon$, which we neglect. We summarise these findings in the following theorem. 
\begin{theorem}[Resource Costs for Polynomial Interpolation]
Let $\tilde{P}_{m}f$ be the $m$-degree polynomial fit taken at the inverse integers closest to the Chebyshev nodes in the interval $[-\width,\width]$.
The resource requirements to achieve error
\begin{align*}
    |\expval{O(T)} - \tilde{P}_mf(0)|\leq \epsilon \norm{O}
\end{align*}
for incoherent and coherent measurement scales as the following table.
\begin{table}[H]
    \centering
    \begin{tabular}{c|c|c}
        \textbf{Scaling}  & \textbf{Max Depth}, $\Dmax$ & \textbf{Total Resources}, $\Cexp$ \\
        \hline
        \textbf{Coherent} & $\tilde{O}\left( \frac{(a_\mathrm{max} \Upsilon \lambda T)^{1+1/p}}{\epsilon}\log^3\big( \frac{1}{\epsilon}\big) \right)$  & $\tilde{O}\left( \frac{(a_\mathrm{max} \Upsilon \lambda T)^{1+1/p}}{\epsilon}\log^3\big( \frac{1}{\epsilon}\big)\right)$  \\
        \hline
        \textbf{Incoherent} & $\tilde{O}\left( (a_\mathrm{max} \Upsilon \lambda T)^{1+1/p}\log^3\big( \frac{1}{\epsilon}\big) \right)$  & $\tilde{O}\left( \frac{(a_\mathrm{max} \Upsilon \lambda T)^{1+1/p}}{\epsilon^2}\log^3\big( \frac{1}{\epsilon}\big) \right)$  
    \end{tabular}
    \label{Table:Resource_Costs_Incoherent_Interpolation}
\end{table}
  \noindent Here, $\lambda\leq \sum_\gamma \norm{H_\gamma}$ is defined in \cref{Lemma:Sufficient_Trotter_Depth}, $\tilde{O}$ hides $\log\log$ factors, and there is a failure probability of $\delta = 0.01$. Moreover, $m = O(\log 1/\epsilon)$.
\end{theorem}

\section{Numerical Demonstrations}
\label{sec:numerical-demonstration}

Here we test the above results using the Heisenberg model on a length $L$ 1D chain, defined as:
\begin{align*}
    H = \sum_{i=1}^{L-1} \left( X_{i}X_{i+1} + Y_{i}Y_{i+1} +Z_{i}Z_{i+1} \right) + \sum_{i=1}^{L} h_i Z_i
\end{align*}
where we choose the values $\{h_i\}_{i=1}^L$ uniformly randomly in the interval $[-1,1]$.
This Hamiltonian is well studied in the context of Trotter simulation \cite{childs2018toward}.
We choose to work with a $p=2$ product formula on a system of $L=6$ qubits.
For the initial state, we choose a randomly chosen classical bit string: $\ket{x}, \ x \in \{0,1\}^n$ and we choose a randomly chosen sum of 3 Pauli strings to act as the observable.
In order to generate the figures in the following section, we have used a minimum number of Trotter steps $r=1/s_{\min} = \left\lceil (\Lambda T)^{3/2} \right\rceil$.

\subsection{Error with Fixed Maximum Circuit Depth}
Here we consider the error associated with choosing a fixed maximum circuit depth that we can utilise (e.g. if we are limited by noise in the physical circuit), and then use this to compare results from time simulation.
That is, suppose we want to predict $\langle O(T) \rangle $ but are restricted to some maximum circuit depth (i.e. Trotter steps).
How does the unextrapolated error compare to the extrapolated error?

\Cref{Fig:Fixed_Time_Change_Degree} (left) shows that if we consider a fixed time and as we increase the number of steps from a minimum value of $r= \left\lceil (\Lambda T)^{3/2} \right\rceil$, using Trotter extrapolation drastically reduces the error compared with just measuring the state directly.
However, as demonstrated by \cref{Fig:Fixed_Time_Change_Degree} (right), the benefits of this are limited by the error associated by the measurements at each point.
As one might expect, neither extrapolation techniques can improve results beyond the limit to which we measure the observable we are extrapolating.

\begin{figure}[h!]
    \centering
    \begin{minipage}{0.45\textwidth}
        \centering
        \includegraphics[width=1.0\textwidth]{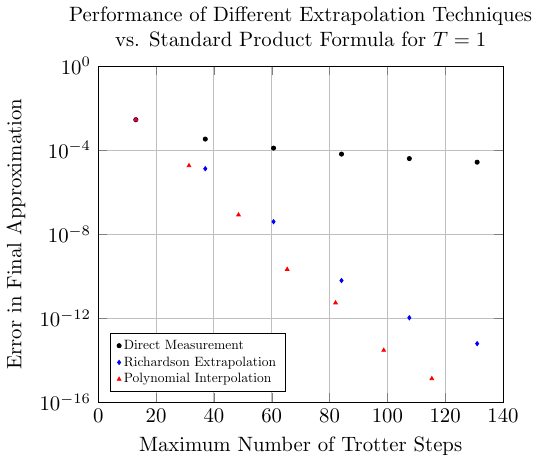}
    \end{minipage}\hfill
    \begin{minipage}{0.47\textwidth}
        \centering
        \includegraphics[width=1.0\textwidth]{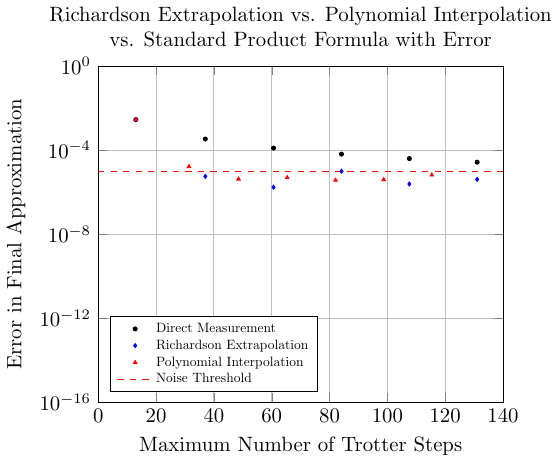} 
    \end{minipage}
    \caption{Left: Error comparison between the observable measured on the time-evolved state using Trotterisation vs. the extrapolated error for different maximum numbers of Trotter steps on a system of $6$ qubits. 
    The initial number of steps corresponds to $\sim (\Lambda T)^{3/2}$.
    We see the plot levels out at the bottom due to floating point precision.
    Right: The same as left, but where the extrapolation procedures are performed with a measurement error of $10^{-6}$ for each measurement.}
    \label{Fig:Fixed_Time_Change_Degree}
\end{figure}

\noindent We can also consider the case where we run the simulation for different amounts of time and see how the error scales with the number of nodes.
From \cref{Fig:Variable_Time} we see that the extrapolation holds for much longer simulation times. 
We point the reader to \cite[Chapter 3]{watkins2024thesis} for similar numerics\footnote{We note a small difference in the scaling of the error compared to \cite{watkins2024thesis} where the gradient changes with $T$. This is due to the different choice in the minimum number of Trotter steps. }.

\begin{figure}[H]
    \centering
    \includegraphics[width=0.6\textwidth]{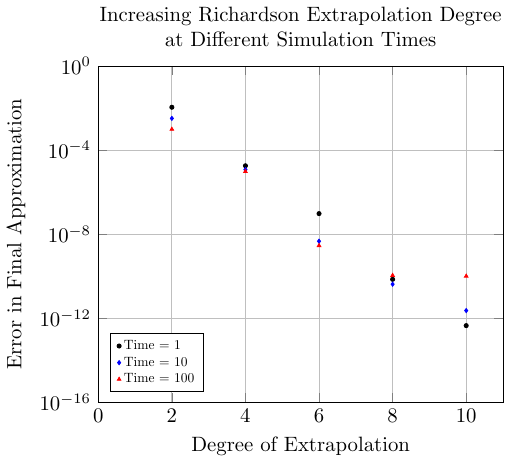}
    \caption{A comparison in the performance of Richardson extrapolation as a function of the Richardson extrapolation degree for different simulation times.
  }
    \label{Fig:Variable_Time}
\end{figure}

\section{Aside: Compatibility with Classical Shadows}
\label{sec:classical_shadows}

In this section, we consider the estimation of many time evolved expectation values using the classical shadows technique. To use classical shadows alongside extrapolation, simply run the classical shadows protocol on each of the $m$ states $\{\rho_k\}_{k=1}^m$ resulting from a Trotter evolution with step size $s_k$. Then compute the expectation values of interest using these shadows and perform the extrapolation.

Classical shadows requires repeated computational basis measurements measurements on the time-evolved states $\rho_k$, which is cheap but limits us to shot-noise accuracy. From \cite{huang2020predicting}, to estimate the expectation value of $M$ local observables $\{O_i\}_{i=1}^M$ to accuracy $|\tr[\rho O_i] - O_{i,est}|\leq \epsdata \norm{O_i}$ with probability $\geq (1-\delta)$, it suffices to take a number of samples of $\rho_k$ which is bounded by
\begin{align*}
	\leq  \frac{128}{\epsdata^2}\max_i\norm{O_i}^2 \log\left( \frac{M}{\delta} \right).
\end{align*}
Each of these samples requires a single Trotter evolution whose cost is determined by $k$, etc. The following theorem is a simple corollary of the above results on Richardson and polynomial extrapolation.
\begin{theorem}
	Let $\{O_i(T)\}_{i=1}^M$ be a set of $M$ $k$-local time-evolved observables $O(T) = e^{iHT}Oe^{-iHt}$.
	To get an estimate $\abs{\tr(\rho O_i(T)) - O_{i,est}}\leq \epsdata \norm{O_i}$ with probability $\geq (1-\delta)$, the maximum Trotter depth scales as
	\begin{align*}
		\Dmax = O\left((\lambda T)^{1+1/p}\polylog(1/\epsilon)\right),
	\end{align*}
	where we need to run experiments at no more than $m = O(\log(1/\epsilon))$ sampling points. Meanwhile, the total number of Trotter steps for the entire protocol scales as
	\begin{align*}
		\Cexp = \tilde{O}\left( \frac{1}{\epsilon^2}(\lambda T)^{1+1/p} \polylog \left(\frac{1}{\epsilon} \right)\log\left(\frac{M}{\delta} \right) \right),
    \end{align*}
    where $\tilde{O}$ suppresses doubly-logarithmic multiplicative factors.
\end{theorem}

\section{Discussion and Conclusions}
\label{sec:discussion-and-conclusions}

In this work, we have shown that by performing product formulae evolutions with different time-step sizes, we can use classical extrapolation techniques to predict the expectation values of observables with gate complexity scaling as $O\left(\frac{T^{1+1/p}}{\epsilon} \polylog\left(\frac{1}{\epsilon}\right)\right)$, where we are ignoring $\log\log$ factors.
We emphasise that the $O(1/\epsilon)$ scaling arises from the measurement error, and that the effective Trotter error scales as $O(\polylog(1/\epsilon))$.

In the NISQ regime, circuit depth is arguably the most important cost metric. For this, we consider extrapolation  using shot-limited measurements of the observable or classical shadows. Here the fact coherent measurements are not needed drastically reduces the constant overheads needed and we see only circuit depths scaling as $O\left( T^{1+1/p} \polylog\left( \frac{1}{\epsilon}\right) \right)$ are required.
We expect the resource scaling with error to be optimal up to $\log$ factors.
The Heisenberg measurement limit is fundamental measurement limit which requires at least $\Omega(1/\epsilon)$ scaling \cite{giovannetti2006quantum}.
Thus we should expect our total resource cost to scale as $\Omega(1/\epsilon)$ in general.

\paragraph{Beyond Standard Product Formulae}
There are many proposals for modifying product formulae, or otherwise combining them with simulations algorithms \cite{rajput2022hybridized,sharma2024hamiltonian}. It remains open whether these algorithms can be fruitfully combined with the Richardson or polynomial extrapolation approaches considered here to improve their performance.

\paragraph{State Dependent Bounds}
The bounds given in this work are also sensitive to the ``physics'' of the system.
That is, if the system has symmetries, or is restricted to a particular subspace, then we can improve the error estimates given here by replacing the spectral norm with a symmetry-respecting norm.
We demonstrate this in \cref{sec:Symmetries}, where we show that if the evolution is restricted to a particular subspace, the various norms or commutators characterising the error can be replaced with quantities projected onto the relevant subspace.

\paragraph{More Tractable Expressions for Commutator Scaling}
Although the results in this paper demonstrate commutator scaling in the error of Richardson extrapolation, actually computing these expressions is likely to be extremely computationally expensive.
Developing an expression which can be computed efficiently is an enormously important task, as it allows us to upper bound the amount of quantum resources needed to reach a guaranteed precision.
We hope that future work will find simplified expressions for the commutator error.

\paragraph{Lower Bounds} The work here gives an improvement on the Trotter error for measured observables.
In the general setting, lower bounds have been proven for the performance of Trotterization methods \cite{hahn2024lower}, but it remains to be seen if similar bounds on the query complexity can be proven if post-processing is allowed. 

\section*{Acknowledgements}

{\begingroup
	\hypersetup{urlcolor=navyblue}
	
	The authors would like to acknowledge very useful discussions with \href{https://orcid.org/0000-0001-7331-2759}{Gyula Lakos} on the convergence of the BCH and Magnus expansions, as well as \href{https://orcid.org/0000-0002-2964-3603}{Dong An}, \href{https://orcid.org/0000-0002-9903-837X}{Andrew Childs}, and \href{https://orcid.org/0000-0002-8334-1120}{Alessandro Roggero} on the topic of Trotter error bounds. 
    The authors would also like to thank \href{https://orcid.org/0009-0006-1500-505X}{Abhishek Rajput} for useful discussions.

	J.D.W. acknowledges support from the United States Department of Energy, Office of Science, Office of Advanced Scientific Computing Research, Accelerated Research in Quantum Computing program, and also NSF QLCI grant OMA-2120757. 
	J.W. acknowledges support from the National Science Foundation under grants DGE-1848739 \& PHY-2310620, as well as the Department of Energy under grant DE-SC0023658. 
	We thank \href{https://www.ectstar.eu/}{ECT*} and the \href{https://fias.institute/en/projects/emmi/}{ExtreMe Matter Institute EMMI at GSI, Darmstadt}, for support in the framework of an ECT*/EMMI Workshop during which this work has been initiated.

\printbibliography[heading=bibintoc]

\endgroup}

\newpage

\appendix

\section{Improved Bounds with Physical Knowledge}
\label{sec:Symmetries}

Here we consider the case where the Hamiltonian may have some physical symmetry. 
In this case, we may be able to exploit this symmetry to improve our convergence bounds.
Let $H =\sum_j h_j$, where $\{h_j\}_j$ are the terms used in the Trotter decomposition.
Then assume that there exists some symmetry subspace denoted $\cS$, such that $\Pi_{\cS}$ is a projector onto the subspace $\cS$.
We assume that the both the overall Hamiltonian and the individual terms in the Trotter decomposition preserve the symmetry.
That is, the following holds.
\begin{align*}
    [H, \Pi_\cS] = 0 \quad \quad  [h_j, \Pi_\cS] =  0
\end{align*}
Consider a state $\ket{\psi_\cS} \in \cS$. 
Then we see that
\begin{align*}
    e^{-ih_jt}\ket{\psi_\cS} &= e^{-ih_jt}\Pi_\cS\ket{\psi_\cS} \\
    &= \Pi_\cS e^{-i\Pi_\cS h_j\Pi_\cS t} \Pi_\cS\ket{\psi_\cS},
\end{align*}
where we have used that $\ket{\psi_\cS} = \Pi_\cS \ket{\psi_\cS}$ and $\Pi_\cS = \Pi_\cS\Pi_\cS$.
Applying this to product formulae,
\begin{align*}
    \cP(t)\ket{\psi_\cS} &= \prod_{j} e^{-ia_jh_j t_j}\ket{\psi_\cS} \\ 
    &= \prod_je^{-ia_j  \Pi_\cS h_j\Pi_\cS t_j}\ket{\psi_\cS}.
\end{align*}
Thus, because the simulation is restricted to a particular subspace, we can consider the Hamiltonian restricted to that subspace.
In this case we no longer need to consider the full operator norm and can consider a ``symmetry projected'' norm.
We can define ``symmetry-respecting'' norms which characterise the rate of convergence.
\begin{align*}
    \Lambda_{ \cS} \coloneqq \sum_j \norm{\Pi_\cS h_j\Pi_\cS} \quad \quad \quad  
\end{align*}
For the case of commutators, from \cref{Lemma:Richardson_Extrapolation_Error},
\begin{align*}
    \lambda_{j,l} &\coloneqq \bigg(\sum_{\substack{j_1\dots j_l\in \sym\mathbb{Z}_+\geq p \\ j_1+\dots+j_l=j}}  \prod_{\kappa=1}^l 2 \frac{\acomm^{(j_\kappa + 1)}}{(j_\kappa + 1)^2}\bigg)^{1/(j+l)}
\end{align*}
and
\begin{align*}
    \alpha^{(j)}_{\text{comm}, \cS} &= \sum_{\gamma_1\gamma_2\ldots \gamma_j = 1}^\Gamma \norm{\Pi_\cS[H_{\gamma_1} H_{\gamma_2} \ldots H_{\gamma_j}]\Pi_\cS}.
\end{align*}
These symmetry respecting quantities can then be used in place of $\Lambda, \acomm^{(j)}$ in \cref{Lemma:Richardson_Extrapolation_Error} and elsewhere, thus giving us improved convergence results.
We believe that similar results should hold when the restriction to the subspace $\cS$ is not strictly preserved, i.e. there is some leakage to other subspaces.
For example, if one restricts to low-energy states as per \cite{csahinouglu2021hamiltonian}.

\section{Polynomial Interpolation: Error from Imperfect Chebyshev Nodes} \label{sec:Imperfect_Chebyshev_Nodes}

As described earlier, if the Cheybshev nodes do not coincide with the inverse integers, we may not be able to sample exactly from the Chebyshev nodes as, if we restrict ourselves to integer applications of the Trotter evolution, we require $s = 1/r, $ for $r\in \mathbb{Z}$.
For any $\xi\in\mathbb{Z}$, the spacing between inverse integers is
\begin{align}
   \left(\frac{1}{\xi } - \frac{1}{\xi +1 } \right)  
    &= \frac{1}{\xi} \sum_{k=1}^\infty \frac{(-1)^k}{\xi^k} \nonumber \\
    &\leq  \frac{1}{\xi} \sum_{k=1}^\infty \frac{1}{\xi^k} \nonumber  \\
    &\leq \frac{2}{\xi^2}, \ \  \text{for sufficiently large $\xi$,} \label{Eq:Inverse_Interger_Spacing}
\end{align}
where one can use the standard bound on a geometric sum to get the last line.
We then use the following theorem.

\begin{theorem}[Perturbed Chebyshev Nodes, Section 3 of \cite{vianello2018stability}] \label{Theorem:Perturbed_Lebesgue_Constant}
    Let $x\in [c,d]$, and let $L_m$ Lebesgue constant of the $m+1$ Chebyshev points in $[c,d]$.
    The Lebesgue constant of the points after they have been perturbed by $\leq \epsilon_N$ away from the Chebyshev node, denoted $L'_m$, is bounded by
    \begin{align*}
        L_m' \leq \frac{L_m}{1-\alpha},
    \end{align*}
    where $\epsilon_m$ satisfies
    \begin{align*}
        \epsilon_m = \frac{\alpha(d-c)}{m^2 L_m}.
    \end{align*}
\end{theorem}

We can apply this to the case where we are not sampling from the exact Chebyshev nodes.
We now consider the interval $[-\width,\width]$ on which we wish to learn the function $f_B(s)$.
Applying \cref{Theorem:Perturbed_Lebesgue_Constant} with $\alpha=1/2$ gives
\begin{align*}
    L'_m \leq 2L_m,
\end{align*}
provided for
\begin{align*}
    \epsilon_m =  \frac{\width}{m^2 L_m}.
\end{align*}
If we consider the interval $[-\width,\width]$, then the maximum distance between inverse point is at the boundaries.
Hence from \cref{Eq:Inverse_Interger_Spacing}, the maximum distance between inverse integers in $[-\width,\width]$ is
\begin{align*}
   \epsilon_m &\leq \max_{\xi \in [-1/\width,1/\width]}  \frac{1}{\xi^2} \\
   &\leq \width^2.
\end{align*}
Thus, for a given $m$, we need to ensure that $\width\leq (m^2L_m)^{-1}$. We note that in general $\width$ scales in terms of $O(T^{1+1/q})$, while $m=O(\log (1/\epsilon))$ and $L_m =O(\log\log(1/\epsilon)$.
To satisfy both bounds we can choose
\begin{align*}
    \width &\leq \frac{1}{(a_\mathrm{max} \Upsilon \lambda T)^{1+1/q} m^2 L_m},  
\end{align*}
and thus the minimum number of Trotter steps scales as
\begin{align*}
  O\left( m/\width \right) =  O\left( (a_\mathrm{max} \Upsilon \lambda T)^{1+1/q} \log^3\left(\frac{1}{\epsilon} \right)\log\log\left( \frac{1}{\epsilon}\right) \right).
\end{align*}

\section{Errata for Previous Literature on Richardson Extrapolation}
\label{sec:Errata}

Here we discuss the previous work on Richardson extrapolation by \citeauthor{vazquez2023well} \cite{vazquez2023well}, which appears to have incorrect derivations for the error in Richardson extrapolation of observables.

In Appendix A, equation (13) of Ref. \cite{vazquez2023well}, it is claimed that the iterated product formula has an expansion
\begin{align} \label{Eq:Incorrect_Expansion_Vazquez}
    \cP^k(T/k) = e^{-iHT} + \sum_{n=1}^\infty A_n \frac{T^{n+1}}{k^n},
\end{align}
where the $A_n$ are size $O(1)$ in $T$ and consist of nested-commutators of the Hamiltonian terms.
However, this expansion is incorrect, which can be seen by comparison to a counter-example\footnote{We add that the series in \cref{Eq:Incorrect_Expansion_Vazquez} (i.e. \cite[eq. (13)]{vazquez2023well}) is claimed to come from Ref. \cite{chin2010multi}, however, we were unable to verify this explicit series appears in  Ref. \cite{chin2010multi}.}.
Consider the Hamiltonian $H=X+Z$ and the associated Trotterization $e^{-iXT/k}e^{-iZT/k}$.
Making the identification $s=1/k$, using the standard BCH formula we get an effective Hamiltonian
\begin{align}
\begin{split}
    \log\left( e^{-iXsT}e^{-iZsT} \right) &=  -isT(X+Z)  -(sT)^2 \frac{1}{2}[X,Z] + \frac{i(sT)^{3}}{12}([X,[X,Z]] + [Z,[Z,X]]) \\  
    &-\frac{(sT)^4}{24}\underbrace{[Z[X,[X,Z]]]}_{=0}   +O(s^5T^5)   \\
    &= -i(X+Z)(sT) + iY(sT)^2 + i(Z+X)\frac{(sT)^3}{3} + O((sT)^5).
\end{split}
\end{align}
We can then exponentiate this and consider the iterated first-order product formula to get
\begin{align}
    \cP^{1/s}(sT) &= \left( e^{-iXsT}e^{-iZsT} \right)^{1/s} \nonumber\\ 
    &= \exp\left(-i(X+Z)T + iYsT^2 + i(Z+X)\frac{s^2T^3}{3} + O(s^4T^5)\right).
\end{align}

To obtain a polynomial expansion in $s$, we use the variation of parameters formula
\begin{align}
    e^{A+B} = e^A + \int_0^1 d\tau e^{(1-\tau)A} B e^{(A+B)\tau}
\end{align}
which we iterate $4$ times to get the following.
\begin{align}
\begin{aligned}
     e^{A+B} &= e^A +\int_0^1 d\tau_1 e^{(1-\tau_1)A} B e^{A\tau_1} \\ 
     &+ \int_0^1 d\tau_1 \int_0^{\tau_1} d\tau_2  e^{(1-\tau_1)A} B e^{A(\tau_1-\tau_2)}B e^{\tau_2A} \\
     &+ \int_0^1 d\tau_1\int_0^{\tau_1} d\tau_2\int_0^{\tau_2} d\tau_3 \ e^{(1-\tau_1)A} B e^{A(\tau_1-\tau_2)}B e^{(\tau_2-\tau_3)A} B e^{\tau_3 A} \\
     &+ \int_0^1 d\tau_1 \int_0^{\tau_1} d\tau_2 \int_0^{\tau_2} d\tau_3\int_0^{\tau_3} d\tau_4 \ e^{(1-\tau_1)A} B e^{A(\tau_1-\tau_2)}B e^{(\tau_2-\tau_3)A} B e^{(\tau_3-\tau_4) A} B  e^{\tau_4 (A+B)} \label{Eq:Higher_Order}
\end{aligned}
\end{align}
We now make the explicit choice
\begin{align}
\begin{aligned}
    A &= -iT(X+Z) \\
    B &= iYsT^2 + i(Z+X)\frac{s^2T^3}{3} + O(s^4T^5)
\end{aligned}
\end{align}
and neglect any terms above $O(s^3)$.
Since $B$ is $O(s)$, then the last line of \cref{Eq:Higher_Order} is $O(s^4)$, hence we neglect it. Keeping some $A$ terms implicit for brevity, we have
\begin{align}
\begin{aligned}
    &\cP^{1/s}(sT) = e^{-iT(X+Z)} \\
    &+ \int_0^1 d\tau_1 e^{A(1-\tau_1)} (iYsT^2 + i(Z+X)\frac{s^2T^3}{3}) e^{A\tau_1}  \\
    &+ \int_0^1 d\tau_1 \int_0^{\tau_1} d\tau_2 \bigg[  e^{A(1-\tau_1)} (iYsT^2 + i(Z+X)\frac{s^2T^3}{3}) e^{A(\tau_1-\tau_2)}(iYsT^2 + i(Z+X)\frac{s^2T^3}{3}) e^{A\tau_2} \bigg]  \\
    &+ \int_0^1 d\tau_1 \int_0^{\tau_1} d\tau_2 \int_0^{\tau_2} d\tau_3 \ e^{A(1-\tau_1)} (iYsT^2) e^{A(\tau_1-\tau_2)}(iYsT^2) e^{A(\tau_2-\tau_3)} (iYsT^2) e^{A\tau_3} + O(s^4).
\end{aligned}
\end{align}
Grouping the above according to combinations of $s$ and $T$, the expression may be written as
\begin{align}
    \cP^{1/s}(sT) &=e^{-iH T} + sT^2G_1(T) + s^2T^3G_2(T) + s^2T^4 G_3(T) + s^3 T^5G_4(T) + s^3 T^6G_5(T) + O(s^4) \label{Eq:CounterExample_2} 
\end{align}
where $G_1,G_2, G_3, G_4, G_5 = O(1)$ are matrix functions of $T$ only.
We note in particular that there are terms with scaling $O(s^2T^4),O(s^3T^6)$, and $O(s^2T^5)$ which do not appear with the ratio $s^mT^{m+1}$ which would be necessary for \cref{Eq:Incorrect_Expansion_Vazquez} to be valid.
It can be checked that this occurs for higher order terms as well.

The incorrect series expansion given in \cref{Eq:Incorrect_Expansion_Vazquez} is then used in Appendix C of Ref. \cite{vazquez2023well} and hence the errors unfortunately propagate into this later section. 
In particular, through equations (15)-(18) which is where the error in the Richardson extrapolation is derived.
Following from \cite[eq. (17)]{vazquez2023well} (and broadly using their notation) with $\cP^{k_j}(T/k_j) = U + E_j$, then
\begin{align*}
    \langle O(k_j^{-1},T) \rangle &= \bra{\psi}U^\dagger O U \ket{\psi} + \bra{\psi}U^\dagger O E_j \ket{\psi} \\
    &+\bra{\psi}E_j^\dagger O U \ket{\psi} + \bra{\psi}E_j^\dagger O E_j \ket{\psi}.
\end{align*}
We now explicitly consider Richardson extrapolation with $l=3$ extrapolation points, for the Hamiltonian $H=X+Z$.
Following \cite[eq. (17)]{vazquez2023well},
\begin{align}
    \epsilon &= \sum_{j=1}^{l=3} a_j  \bra{\psi}U^\dagger O E_j \ket{\psi} \\
     &= \sum_{j=1}^{l=3} a_j \bra{\psi}U^\dagger O \left( T^2\frac{1}{k_j}G_1 + \frac{1}{k_j^2}T^3G_2 + \frac{1}{k_j^2}T^4G_3 + \frac{1}{k_j^3}T^5G_4 + \frac{1}{k_j^3}T^6G_5  \right)  \ket{\psi} + \cO\left( \frac{1}{k_1^4}\right).
\end{align}
We see that the Richardson conditions cancels off all terms up to order $O(k_1^{-2})$. 
This gives
\begin{align*}
    \epsilon = \sum_{j=1}^{l=3} a_j \bra{\psi}U^\dagger O \left(  \frac{1}{k_j^3}T^4G_4 + \frac{1}{k_j^3}T^6G_5  \right)  \ket{\psi} + \cO\left( \frac{1}{k_1^4}\right).
\end{align*}
Thus $\epsilon = \cO\left(\frac{T^6}{k_1^3}\right)$ is the asymptotic scaling of the error.
This gives an overall error scaling of the Richardson estimator as
\begin{align}\label{Eq:Explicit_Richardson_Error}
    \sum_{j=1}^{l=3} a_j\langle O(k_j^{-1},T) \rangle = \bra{\psi}U^\dagger O U \ket{\psi} + \cO\left(\frac{T^6}{k_1^3}\right).
\end{align}

The result is that the asymptotic behaviour of the error given in equation (7) in the main text of Ref. \cite{vazquez2023well} appears to be incorrect.
The correct version of \cref{Eq:Incorrect_Expansion_Vazquez} appears in \cite[Lemma 7]{aftab2024multi}.
We also realise that for the error in \cref{Eq:Explicit_Richardson_Error} to shrink, we require that $k_1 = O(T^2),$ and hence we would require a number of Trotter steps (and hence circuit depths) scaling as $O(T^2)$. 
As such, the error scaling in \cref{Eq:Explicit_Richardson_Error} is consistent with the bounds present in the current work, as seen in \cref{Theorem:Informal_Main_Theorem} where the time scaling is $O(T^2)$ for $p=1$ product formulae.

\end{document}